\documentclass[a4paper,UKenglish,cleveref,autoref,thm-restate]{lipics-v2021}

\title{Compression by Contracting Straight-Line Programs}
\titlerunning{Compression by Contracting Straight-Line Programs}

\author{Moses Ganardi}
{Max Planck Institute for Software Systems (MPI-SWS), Kaiserslautern, Germany}
{ganardi@mpi-sws.org}
{https://orcid.org/0000-0002-0775-7781}
{}

\authorrunning{Moses Ganardi}
\Copyright{Moses Ganardi}



\acknowledgements{The author thanks Pawe\l{} Gawrychowski, Artur Je\.z, Philipp Reh, and Louisa Seelbach Benkner for helpful discussions.
The author is also indebted to the anonymous referees whose comments improved the presentation of this work.}

\nolinenumbers

\hideLIPIcs

\EventEditors{Petra Mutzel, Rasmus Pagh, and Grzegorz Herman}
\EventNoEds{3}
\EventLongTitle{29th Annual European Symposium on Algorithms (ESA 2021)}
\EventShortTitle{ESA 2021}
\EventAcronym{ESA}
\EventYear{2021}
\EventDate{September 6--8, 2021}
\EventLocation{Lisbon, Portugal}
\EventLogo{}
\SeriesVolume{204}
\ArticleNo{24}

\usepackage{stmaryrd}
\usepackage{tikz,tikz-qtree}
\usetikzlibrary{arrows,arrows.meta,automata,backgrounds,calc,decorations.pathmorphing,fit,positioning,shapes,trees}

\newcommand{\N}{\mathbb{N}}
\newcommand{\Z}{\mathbb{Z}}
\newcommand{\B}{\mathcal{B}}

\newcommand{\F}{\mathcal{F}}
\newcommand{\G}{\mathcal{G}}
\renewcommand{\H}{\mathcal{H}}
\renewcommand{\L}{\mathcal{L}}
\renewcommand{\O}{\mathcal{O}}
\newcommand{\R}{\mathcal{R}}
\newcommand{\Q}{\mathcal{Q}}
\renewcommand{\S}{\mathcal{S}}

\newcommand{\V}{\mathcal{V}}

\newcommand{\val}[1]{\llbracket #1 \rrbracket}
\newcommand{\rk}{\mathsf{rk}}
\renewcommand{\dag}{\mathsf{dag}}
\newcommand{\lft}{\mathsf{left}}

\newcommand{\lsib}{\mathsf{lsib}}
\newcommand{\level}{\mathsf{level}}
\newcommand{\height}{\mathsf{height}}
\newcommand{\rev}{\mathsf{rev}}
\renewcommand{\root}{\mathsf{root}}

\newcommand{\access}{\mathsf{access}}
\newcommand{\setfinger}{\mathsf{setfinger}}
\newcommand{\movefinger}{\mathsf{movefinger}}

\renewcommand{\insert}{\mathsf{insert}}
\newcommand{\delete}{\mathsf{delete}}
\newcommand{\pred}{\mathsf{pred}}
\newcommand{\suc}{\mathsf{succ}}
\newcommand{\rank}{\mathsf{rank}}
\newcommand{\select}{\mathsf{select}}
\renewcommand{\split}{\mathsf{split}}

\begin{document}

\maketitle

\begin{abstract}
	In grammar-based compression a string is represented by a context-free grammar, also called a {\em straight-line program (SLP)},
	that generates only that string.
	We refine a recent balancing result stating that one can transform an SLP of size $g$ in linear time
	into an equivalent SLP of size $\O(g)$ so that the height of the unique derivation tree is $\O(\log N)$
	where $N$ is the length of the represented string (FOCS 2019).
	We introduce a new class of balanced SLPs, called {\em contracting} SLPs,
	where for every rule $A \to \beta_1 \dots \beta_k$ the string length of every variable $\beta_i$ on the right-hand side is smaller
	by a constant factor than the string length of $A$.
	In particular, the derivation tree of a contracting SLP has the property that every subtree
	has logarithmic height in its leaf size.
	We show that a given SLP of size $g$ can be transformed in linear time into an equivalent contracting SLP
	of size $\O(g)$ with rules of constant length.
	This result is complemented by a lower bound, proving that converting SLPs into
	so called $\alpha$-balanced SLPs or AVL-grammars can incur an increase by a factor of $\Omega(\log N)$.

	We present an application to the navigation problem in compressed unranked trees, represented by {\em forest straight-line programs (FSLPs)}.
	A linear space data structure by Reh and Sieber (2020) supports navigation steps such as going to the parent, left/right sibling,
	or to the first/last child in constant time.
	We extend their solution by the operation of moving to the $i$-th child in time $\O(\log d)$
	where $d$ is the degree of the current node.

	Contracting SLPs are also applied to the finger search problem over SLP-compressed strings
	where one wants to access positions near to a pre-specified {\em finger} position,
	ideally in $\O(\log d)$ time where $d$ is the distance between the accessed position and the finger.
	We give a linear space solution for the dynamic variant where one can set the finger in $\O(\log N)$ time,
	and then access symbols or move the finger in time $\O(\log d + \log^{(t)} N)$ for any constant $t$
	where $\log^{(t)} N$ is the $t$-fold logarithm of $N$.
	This improves a previous solution by Bille, Christiansen, Cording, and G{\o}rtz (2018) with access/move time $\O(\log d + \log \log N)$.
\end{abstract}

\newpage

\section{Introduction}

In grammar-based compression a long string is represented by a context-free grammar, also called a {\em straight-line program (SLP)},
that generates only that string.
Straight-line programs can achieve exponential compression, e.g.\ a string of length $2^n$ can be produced
by the grammar with the rules $A_n \to A_{n-1} A_{n-1}, \dots, A_0 \to a$.
While it is $\mathsf{NP}$-hard to compute a smallest SLP for a given string \cite{CharikarLLPPSS05}
there are efficient grammar-based compressors of both practical and theoretical interest such as
the LZ78/LZW-algorithms \cite{Ziv1978,Welch84}, \textsc{Sequitur} \cite{Nevill-ManningW97}, and \textsc{Re-Pair} \cite{Larsson2000}.
There is a close connection between grammar-based compression and the LZ77 algorithm, which parses a string into $z$ phrases (without self-references):
On the one hand $z$ is always a lower bound on the size of the smallest SLP for the string \cite{CharikarLLPPSS05}.
On the other hand one can always construct from the LZ77 parse an SLP of size $\O(z \log N)$ where $N$ is the string length
\cite{CharikarLLPPSS05,Rytter03} (see also \cite{Gawrychowski11} for LZ77 with self-referential phrases).
Furthermore, the hierarchical structure of straight-line programs
makes them amenable to algorithms that work directly on the compressed representation,
without decompressing the string first.
We refer to \cite{Lohrey12} for a survey on the broad literature on algorithms on grammar-compressed data.

\subsection{Balanced grammars}

For some algorithmic applications it is useful if the SLP at hand satisfies certain balancedness conditions.
In the following we always denote by $N$ the length of the represented string.
A recent result states that one can transform an SLP of size $g$ in linear time
into an equivalent SLP of size $\O(g)$ so that the height of the unique derivation tree is $\O(\log N)$ \cite{GanardiJL19}.
This yields a clean $\O(g)$ space data structure which supports random access to any position $i$ in the string in time $\O(\log N)$,
by descending in the derivation tree from the root to the $i$-th leaf.
The original solution for the random access problem by Bille, Landau, Raman, Sadakane, Satti, and Weimann relied on a sophisticated
weighted ancestor data structure \cite{BilleLRSSW15}.
Its advantage over the balancing approach from \cite{GanardiJL19} is that
it supports random access to the string defined by any given variable $A$ in time $\O(\log |A|)$.

Although the derivation tree of an SLP may have logarithmic height its subtrees may still be very unbalanced.
Arguably, the strongest balancedness notions are $\alpha$-balanced SLPs introduced by Charikar, Lehman, Liu, Panigrahy, Prabhakaran, Sahai, and Shelat \cite{CharikarLLPPSS05}
and AVL-grammars proposed by Rytter \cite{Rytter03}.
An SLP in Chomsky normal form is {\em $\alpha$-balanced} if for every rule $A \to BC$
the ratios $|B|/|A|$ and $|C|/|A|$ lie between $\alpha$ and $1-\alpha$ where $0 < \alpha \le 1/2$ is some constant.
An {\em AVL-grammar} is again an SLP in Chomsky normal form whose derivation tree is an AVL-tree,
i.e.\ for every rule $A \to BC$ the heights of the subtrees below $B$ and $C$ differ at most by one.
In fact, the aforementioned transformations from LZ77 into SLPs produce an $\alpha$-balanced SLP, with $\alpha \le 1-\frac{1}{2}\sqrt{2}$,
and an AVL-grammar, respectively \cite{CharikarLLPPSS05,Rytter03}.
Using the same proof techniques one can also transform an SLP of size $g$ into an $\alpha$-balanced SLP
or an AVL-grammar of size $\O(g \log N)$ \cite{CharikarLLPPSS05,Rytter03}.

Let us list a few algorithmic results on $\alpha$-balanced SLPs and AVL-grammars.
Note that in the following bounds we can always replace $g$ by $\O(g \log N)$ and allow arbitrary SLPs as an input.
Gawrychowski \cite{Gawrychowski11} proved that the {\em compressed pattern matching problem},
given an SLP of size $g$ for a string $s$ and a pattern $m$, does $p$ occur in $s$,
can be solved in $\O(g+m)$ time, assuming that the SLP is $\alpha$-balanced.
Gagie, Gawrychowski, Kärkkäinen, Nekrich, and Puglisi \cite{GagieGKNP12}
presented a solution for the {\em bookmarking problem} in $\alpha$-balanced SLPs or AVL-grammars of size $g$.
Given $b$ positions in the string, called bookmarks, we can decompress any substring of length $\ell$
that covers a bookmark in time $\O(\ell)$ and space $\O(g+b \log^* N)$.
Based on this bookmarking data structure they present {\em self-indexes} for LZ77 and SLPs \cite{GagieGKNP12,GagieGKNP14},
which support extracting substrings and finding all occurrences of a given pattern.
Finally, we mention the solutions for the Hamming distance problem and the subsequence problem on SLP-compressed strings,
considered by Abboud, Backurs, Bringmann, and Künnemann \cite{AbboudBBK17}.
As a first step their algorithms convert the input SLPs into AVL-grammars, and solve both problems in time
$\tilde \O(g^{1.410} \cdot N^{0.593})$, improving on the decompress-and-solve $\O(N)$ time algorithms.

\subsection{Main results}

The starting point of this paper is the observation that the size increase by a $\O(\log N)$ factor
in the transformation from SLPs to $\alpha$-balanced SLPs or AVL-grammars is unavoidable (\Cref{thm:lb}).
This lower bound holds whenever in the derivation tree any path from a variable $A$ to a leaf has length $\Theta(\log |A|)$.
This motivates the search for balancedness notions of SLPs that can be established
without increasing the size by more than a constant factor
and that provide good algorithmic properties.
We introduce a new class of balanced SLPs, called {\em contracting straight-line programs},
in which every variable $\beta_i$ occurring on the right-hand side of a rule $A \to \beta_1 \dots \beta_k$
satisfies $|\beta_i| \le |A|/2$.
The derivation tree of an contracting SLP has the property that every subtree has logarithmic height in its leaf size,
i.e.\ in the number of descendant leaves.
We explicitly admit rules with right-hand sides of length greater than two, however, the length will always
be bounded by a constant in this paper.
We say that an SLP $\G$ defines a string $s$ if some variable in $\G$ derives $s$ (and $s$ only).
The main theorem of this paper refines the balancing theorem from \cite{GanardiJL19} as follows:

\begin{theorem}
	\label{thm:contracting}
	Given an SLP $\G$ of size $g$, one can compute in linear time a contracting SLP of size $\O(g)$
	with constant-length right-hand sides which defines all strings that $\G$ defines.
\end{theorem}

As an immediate corollary we obtain a simple $\O(g)$ size data structure which supports random access
to any symbol $A$ of the SLP in time $\O(\log |A|)$ instead of $\O(\log N)$.
This is useful whenever multiple strings $s_1, \dots, s_m$ are compressed using a single SLP
since we can support random access to any string $s_i$ in time $\O(\log |s_i|)$.
We present an example application to unranked trees represented by {\em forest straight-line programs (FSLPs)} introduced in \cite{GasconLMRS20}.
FSLPs are a natural generalization of string SLPs that can compress trees both horizontally and vertically,
and share the good algorithmic applicability of their string counterparts \cite{GasconLMRS20}.
Reh and Sieber presented a linear space data structure on FSLP-compressed trees that allows to perform navigation steps in constant time,
such as moving to the first/last child, left/right sibling, parent node, and returning the symbol of the current node \cite{RehS20}.
Using contracting SLPs we can extend their data structure by the operation of moving to the $i$-th child,
for a given number $1 \le i \le d$, in time $\O(\log d)$ where $d$ is the degree of the current node.

\begin{theorem}
	\label{thm:fslp}
	Given an FSLP $\G$ of size $g$,
	one can compute an data structure in $\O(g)$ time and space
	supporting the following operations in constant time:
	Move to the root of the first/last tree of a given variable,
	move to the first/last child, to the left/right sibling or to the parent of the current node,
	return the symbol of the current node.
	One can also move to the $i$-th child of the current node in time $\O(\log d)$
	where $d$ is the degree of the current node.
\end{theorem}

A second application concerns the {\em finger search problem} on grammar-compressed strings.
A finger search data structure supports fast updates and searches to elements
that have small rank distance from the {\em fingers}, which are pointers to elements in the data structure.
The survey \cite{Brodal04} provides a good overview on dynamic finger search trees.
In the setting of finger search on a string $s$, Bille, Christiansen, Cording, and G{\o}rtz \cite{BilleCCG18} considered three operations:
$\access(i)$ returns symbol $s[i]$, $\setfinger(i)$ sets the finger at position $i$ of $s$,
and $\movefinger(i)$ moves the finger to position $i$ in $s$.
Given an SLP of size $g$ for a string of length $N$, they presented an $\O(g)$ size data structure
which supports $\setfinger(i)$ in $\O(\log N)$ time,
and $\access(i)$ and $\movefinger(i)$ in $\O(\log d + \log \log N)$ time where $d$ is the distance
from the current finger position \cite{BilleCCG18}.
If we assume that the SLP is $\alpha$-balanced or an AVL-grammar, one can come up with a linear space solution which supports
$\access(i)$ and $\movefinger(i)$ in $\O(\log d)$ time (\Cref{thm:opt}).
For general SLPs we present a finger search structure with improved time bounds:

\begin{theorem}
	\label{thm:finger}
	Let $t \ge 1$. Given an SLP of size $g$ for a string of length $N$,
	one can support $\setfinger(i)$ in $\O(\log N)$ time, and $\access(i)$ and $\movefinger(i)$ in $\O(\log d + \log^{(t)} N)$ time,
	where $d$ is the distance between $i$ and the current finger position, after $\O(tg)$ preprocessing time and space.
\end{theorem}

Here $\log^{(t)} N$ is the $t$-fold logarithm of $N$, i.e. $\log^{(0)} N = N$ and $\log^{(t+1)} N = \log \log^{(t)} N$.
Choosing any constant $t$ we obtain a linear space solution for dynamic finger search,
supporting $\access(i)$ and $\movefinger(i)$ in $\O(\log d + \log^{(t)} N)$ time.
Alternatively, we obtain a clean $\O(\log d)$ time solution if we admit a $\O(g \log^* N)$ space data structure.
Furthermore, \Cref{thm:finger} also works for multiple fingers where every finger uses additional $\O(\log N)$ space.

Let us remark that \Cref{thm:contracting} holds in the {\em pointer machine model} \cite{Tarjan79},
whereas for \Cref{thm:fslp} and \Cref{thm:finger} we assume the {\em word RAM model}
with the standard arithmetic and bitwise operations on $w$-bit words, where $w \ge \log N$.
The assumption on the word length is standard in the area of grammar-based compression, see \cite{BilleLRSSW15,BilleCCG18}.

\subsection{Overview of the proofs}

The proof of \Cref{thm:contracting} follows the ideas from \cite{BilleLRSSW15} and \cite{GanardiJL19}.
The obstacle for $\O(\log N)$ time random access or $\O(\log N)$ height are occurrences of {\em heavy} variables on right-hand sides
of rules $A \to \beta_1 \dots \beta_k$, i.e.\ variables $\beta_i$ whose length exceeds $|A|/2$.
These occurrences can be summarized in the {\em heavy forest},
which is a subgraph of the directed acyclic graph associated with the SLP.
The random access problem can be reduced to {\em weighted ancestor queries} (see \Cref{sec:ds})
on every heavy tree whose edges are weighted by the lengths of the variables
that branch off from the heavy tree.
Using a ``biased'' weighted ancestor data structure
one can descend in the derivation tree in $\O(\log N)$ time, spending amortized constant time on each heavy tree \cite{BilleLRSSW15}.
Our main contribution is a solution of the weighted ancestor problem in the form of an SLP:
Given a tree $T$ of size $n$ where the edges are labeled by weighted symbols, we construct a contracting SLP of size $\O(n)$
defining all prefixes in $T$, i.e.\ labels of paths from the root to some node.
The special case of defining all prefixes of a weighted string by a weight-balanced SLP of linear size (i.e.\ $T$ is a path)
was solved in \cite{GanardiJL19}; however, the constructed SLP only satisfies a weaker balancedness condition.

To solve finger search efficiently, Bille, Christiansen, Cording, and G{\o}rtz first consider the {\em fringe access problem} \cite{BilleCCG18}:
Given a variable $A$ and a position $1 \le i \le |A|$, access symbol $A[i]$, ideally in time $\O(\log d)$ where $d = \min\{i, |A|-i+1\}$.
For this purpose the SLP is partitioned into {\em leftmost} and {\em rightmost trees},
which produce strings of length $N$, $N^{1/2}$, $N^{1/4}$, $N^{1/8}$, etc.
The leftmost/rightmost trees can be traversed in $\O(\log \log N)$ time using a $\O(\log \log N)$ time
weighted ancestor data structure by Farach-Colton and Muthukrishnan \cite{FarachM96}.
Applying this approach to contracting SLPs one can solve fringe access in time $\O(\log d + \log \log |A|)$
since the trees have $\O(\log N)$ height, for which one can answer weighted ancestor queries in constant time
using a predecessor data structure by P\v{a}tra\c{s}cu-Thorup \cite{PatrascuT14}.
Using additional weighted ancestor structures, we can reduce the term $\log \log |A|$ to $\log^{(t)} N$.

\section{Straight-line programs}

\label{sec:prelim}

A {\em context-free grammar} $\G = (\V,\Sigma,\R,S)$ consists of a finite set $\V$ of {\em variables},
an alphabet $\Sigma$ of {\em terminal symbols}, where $\V \cap \Sigma = \emptyset$,
a finite set $\R$ of {\em rules} $A \to u$ where $A \in \V$ and $u \in (\V \cup \Sigma)^*$ is a {\em right-hand side},
and a {\em start variable} $S \in \V$.
The set of {\em symbols} is $\V \cup \Sigma$.
We call $\G$ a {\em straight-line program (SLP)} if every variable occurs exactly once on the left-hand side of a rule
and there exists a linear order $<$ on $\V$ such that $A < B$ whenever $B$ occurs
on the right-hand side of a rule $A \to u$.
This ensures that every variable $A$ derives a unique string $\val{A} \in \Sigma^*$.
We also write $|A|$ for $|\val{A}|$.
A string $s \in \Sigma^*$ is {\em defined} by $\G$ if $\val{A} = s$ for some $A \in \V$.
The {\em size} of $\G$ is the total length of all right-hand sides of the rules in $\G$.
For a variable $A \in \V$ we denote by $\height(A)$ the height of the unique derivation tree rooted in $A$.
The height of $\G$ is the height of $S$.
We define the directed acyclic graph $\dag(\G) = (\V \cup \Sigma, E)$
where $E$ is a multiset of edges, containing for every rule $(A \to \beta_1 \dots \beta_k) \in \R$ with $\beta_1, \dots, \beta_k \in \V \cup \Sigma$
the edges $(A,\beta_1), \dots, (A,\beta_k)$.
An SLP $\G$ can be transformed in linear time into an SLP $\G'$ in {\em Chomsky normal form} which defines all strings that $\G$ defines,
i.e.\ each rule is either a binary rule $A \to BC$ or a terminal rule $A \to a$ where $A,B,C \in \V$ and $a \in \Sigma$.

An SLP is {\em $\alpha$-balanced}, for some constant $0 < \alpha \le 1/2$, if it is in Chomsky normal form
and for all rules $A \to BC$ both $|B|/|A|$ and $|C|/|A|$ lie between $\alpha$ and $1-\alpha$.
An {\em AVL-grammar} is an SLP in Chomsky normal form where for all rules $A \to BC$
we have $|\height(B)-\height(C)| \le 1$.
An SLP in Chomsky normal form is {\em $(\alpha,\beta)$-path balanced}, for some constants $0 < \alpha \le \beta$, if for every variable $A$
the length of every root-to-leaf path in the derivation tree is between $\alpha \log |A|$ and $\beta \log |A|$.
Observe that every $\alpha$-balanced SLP is $(1/\log (\alpha^{-1}),1/\log ((1-\alpha)^{-1}))$-path balanced
and AVL-grammars are $(0.5,2)$-path balanced.
The latter follows from the fact that the height decreases at most by 2 when going from an AVL-tree
to an immediate subtree.
There are algorithms that compute for given a string $w$ an $\alpha$-balanced SLP \cite{CharikarLLPPSS05} and an AVL-grammar \cite{Rytter03}
of size $\O(g \log N)$ where $g$ is the size of a smallest SLP for $w$.
We show that these bounds are optimal even for path balanced SLPs: There are strings for which the smallest path balanced SLPs
have size $\Omega(g \log N)$.

\begin{theorem}
	\label{thm:lb}
	There exists a family of strings $(s_n)_{n \ge 1}$ over $\{a,b\}$ such that $|s_n| = \Omega(2^n)$, $s_n$ has an SLP of size $\O(n)$
	and every $(\alpha,\beta)$-path balanced SLP has size $\Omega(n^2)$.
\end{theorem}

\begin{proof}
First we use an unbounded alphabet.
Let $s_n = b_1 a^{2^n} b_2 a^{2^n} \dots b_{n-1} a^{2^n} b_n$, which has an SLP of size $\O(n)$
with the rules $S \to b_1 A_n b_2 A_n \dots b_n$, $A_0 \to a$ and $A_i \to A_{i-1} A_{i-1}$ for all $1 \le i \le n$.
Consider an $(\alpha,\beta)$-path balanced SLP $\G$ for $s_n$.
We will show that $\dag(\G)$ has $\Omega(n^2)$ edges.
Let $1 \le i \le n$ and consider the unique path in $\dag(\G)$ from the starting variable to $b_i$.
Let $\pi_i$ be the suffix path starting in the lowest node $A_i$ such that $\val{A_i}$ contains some symbol $b_j$ with $i \neq j$.
Therefore $|A_i| \ge 2^n$.
Since $\G$ is $(\alpha,\beta)$-path balanced $\pi_i$ has length $\ge \alpha n$.
Since all paths $\pi_i$ are edge-disjoint it follows that $\G$ has size $\Omega(n^2)$.

For a binary alphabet define the separator string $T_i = b a^{2i-2} b a^{2i-1} b$ for $1 \le i \le n$
and define $s_n = T_1 a^{2^n} T_2 \dots T_{n-1} a^{2^n} T_n$ of length $\Omega(2^n)$.
The string $s_n$ has an SLP of size $\O(n)$, which first defines all strings $a^0, \dots, a^{2n-1}$
and then all separator strings $T_i$.
Consider an $(\alpha,\beta)$-path balanced SLP $\G$ for $s_n$.
To prove that $\dag(\G)$ has $\Omega(n^2)$ edges
we use the fact that each substring $b a^k b$ for $0 \le k \le 2n-1$ occurs exactly once in $s_n$ since $2n-1 < 2^n$.
Let $1 \le i \le n$ and consider the unique path $\rho_i$ in $\dag(\G)$ from the starting variable to the symbol $b$
in the middle of the separator $T_i = b a^{2i-2} b a^{2i-1} b$.
Let $B_i$ be the lowest node on $\rho_i$ such that $\val{B_i}$ contains either $b a^{2i-2} b$ or $b a^{2i-1} b$.
Since the successor of $B_i$ on $\rho_i$ produces a string strictly shorter than $|T_i| \le 4n$,
the suffix path of $\rho_i$ starting in $B_i$ has length at most $1+\beta \log (4n) = \O(\log n)$.
Let $A_i$ be the lowest ancestor of $B_i$ on $\rho_i$ such that $\val{A_i}$ contains some symbol from a separator $T_j$ for $i \neq j$.
Therefore $|A_i| \ge 2^n$ and hence the suffix path of $\rho_i$ starting in $A_i$ has length at least $\alpha \log(2^n) = \alpha n = \Omega(n)$.
This implies that the path $\pi_i$ from $A_i$ to $B_i$ has length $\Omega(n) - \O(\log n) = \Omega(n)$.
All paths $\pi_i$ are edge-disjoint since for any edge $(X,Y)$ in $\pi_i$, $\val{Y}$ is of the form
$a^\ell b a^{2i-2} b a^r$ or $a^\ell b a^{2i-1} b a^r$.
This implies that $\G$ has size $\Omega(n^2)$. 
\end{proof}

We will define contracting SLPs with respect to a weighted alphabet.
A {\em weighted alphabet} is a finite set of symbols $\Gamma$ equipped with a weight function
$\|\cdot\| \colon \Gamma \to \N \setminus \{0\}$,
which is extended additively to $\Gamma^*$.
The standard weight function is the length function $| \cdot |$.
A symbol $\beta$ occurring in a weighted string $s$ is {\em heavy in $s$} if $\|\beta\| > \|s\|/2$;
otherwise it is {\em light in $s$}.

Consider an SLP $\G = (\V,\Sigma,\R,S)$ over a weighted alphabet $\Sigma$.
We extend the weight function $\|\cdot\|$ to $\V$ by $\|A\| = \|\val{A}\|$.
A symbol $\beta \in \V \cup \Sigma$ is a {\em heavy child} of $A \in \V$
if $\beta$ is heavy on the right-hand side of the rule $A \to u$.
We also call $\beta$ a {\em heavy} symbol.
A rule $A \to u$ is {\em contracting} if $u$ contains no heavy variables,
i.e.\ every variable $B$ occurring in $u$ satisfies $\|B\| \le \|A\|/2$.
Let us emphasize that heavy {\em terminal} symbols from $\Sigma$ are permitted in contracting rules.
If all rules in $\G$ are contracting we call $\G$ contracting.
By {\em expanding} an occurrence of a variable $B$ in a rule $A \to uBv$
we mean replacing that occurrence of $B$ by the right-hand side $x$ of its rule $B \to x$.
If $B$ is heavy in $uBv$ and $B \to x$ is a contracting rule then the new rule $A \to uxv$ is contracting.

\section{Transformation into contracting SLPs}

\label{sec:prefix}

A {\em labeled tree} $T = (V,E,\gamma)$ is a rooted tree where each edge $e \in E$ is labeled by a string $\gamma(e)$
over a weighted alphabet $\Gamma$.
A {\em prefix} in $T$ is the labeling of a path in $T$ from the root to some node.
The first step towards proving \Cref{thm:contracting} is a reduction to the following problem:
Given a labeled tree $T$, construct a contracting SLP over the weighted alphabet $\Gamma$ of size $\O(|T|)$ which defines all prefixes in $T$.

\subsection{Decomposition into heavy trees}

Consider an SLP $\G = (V,\Sigma,\R,S)$
and suppose that a rule $A \to \beta_1 \dots \beta_k$ contains a unique heavy symbol $\beta_i$.
Then we call $\beta_1 \dots \beta_{i-1}$ the {\em light prefix} of $A$ and 
$\beta_{i+1} \dots \beta_k$ the {\em light suffix} of $A$.
The {\em heavy forest} $H = (\V \cup \Sigma,E_H)$ contains all edges
$(A,\beta)$ where $\beta \in \V \cup \Sigma$ is a heavy child of $A \in \V$,
which is a subgraph of $\dag(\G)$.
Notice that the edges in $H$ point towards the roots, i.e.\ if $(\alpha,\beta) \in E$
then $\alpha$ is a child of $\beta$ in $H$.
We define two labeling functions:
The {\em left label} $\lambda(e)$ of an edge $e = (A,\beta)$ is the {\em reversed} light prefix of $A$
and the {\em right label} $\rho(e)$ of $e$ is the light suffix of $A$.
The connected components of $(H,\lambda)$ and $(H,\rho)$ are called the
{\em left labeled} and {\em right labeled heavy trees},
which can be computed in linear time from $\G$.
The root of a heavy tree is either a terminal symbol
or a variable whose rule is contracting.

\begin{figure}

\centering

\begin{tikzpicture}

\tikzset{every node/.style={minimum size = 4pt, inner sep = 2pt, circle}}
\tikzstyle{heavy}=[draw=orange,thick]

\node (X) {$S$};
\node[below right = 1cm and 0cm of X] (Y) {$T$};
\node[below left = 1cm and 1cm of X] (A) {$A$};
\node[below right = 1cm and 1.5cm of X] (B) {$B$};
\node[below left = 1cm and 1cm of Y] (C) {$C$};
\node[below left = 1cm and 0cm of Y] (D) {$D$};
\node[below right = 1cm and 1cm of Y] (Z) {$U$};
\node[below right = 1cm and 2.5cm of Y] (E) {$E$};
\node[above right = 1cm and -.5cm of E] (W) {$V$};

\draw[->,>=stealth,heavy] (X) to node[left] {\scriptsize $A$} node[right] {\scriptsize $B$} (Y);
\draw[->,>=stealth] (X) to (A);
\draw[->,>=stealth] (X) to (B);
\draw[->,>=stealth] (Y) to (C);
\draw[->,>=stealth] (Y) to (D);
\draw[->,>=stealth,heavy] (Y) to node[below left] {\scriptsize $DC$} node[above right] {\scriptsize $\varepsilon$}  (Z);
\draw[->,>=stealth,heavy] (W) to node[above left] {\scriptsize $\varepsilon$} node[below right] {\scriptsize $E$}  (Z);
\draw[->,>=stealth] (W) to (E);

\end{tikzpicture}

\caption{An excerpt from the dag representation of an SLP. The variables $S,T,U,V$ form a heavy tree with root $U$.
The value of $S$ can be split into the prefix $ACD$, the root $U$ of the heavy tree, and the suffix $B$.
Observe that the left labeling of the path from $U$ to $S$ is $DCA$, which is the reverse of the prefix $ACD$.}
\label{fig:heavy-forest}

\end{figure}
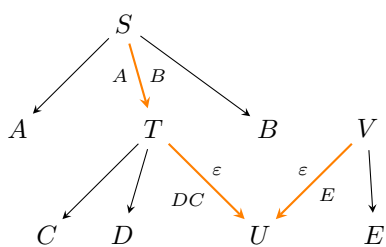

\begin{proposition}
	\label{prop:reduction-to-trees}
	Given an SLP $\G$
	and contracting SLPs $\H_{\mathsf{L}}$ and $\H_{\mathsf{R}}$
	defining all prefixes of all left labeled and right labeled heavy trees of $\G$.
	Let $g$ be the total number of variables in the SLPs and $r$ be the maximal length of a right-hand side.
	One can compute in linear time a contracting SLP $\G'$ which defines all strings that $\G$ defines,
	has $\O(g)$ variables and right-hand sides of length $\O(r)$.
\end{proposition}

\begin{proof}
	Given an SLP $\G = (\V,\Sigma,\R,S)$.
	Consider a variable $A \in \V$ and its unique derivation
	\begin{equation}
		\label{eq:deriv}
		A = \beta_0 \Rightarrow_\G u_1 \beta_1v_1 \Rightarrow_\G u_1u_2\beta_2v_2v_1 \Rightarrow_\G \dots \Rightarrow_\G u_1 \dots u_m \beta_m v_m \dots v_1
	\end{equation}
	where in every step we apply the rule $\beta_{i-1} \to u_i \beta_i v_i$ and $\beta_i$ is the parent node of $\beta_{i-1}$ in $H$.
	Finally, $\beta_m$ is the root of the heavy tree containing $A$, which may be a variable or a terminal symbol.
	The SLP $\H_{\mathsf{R}}$ contains a variable $S_A$ with $\val{S_A} = v_m \dots v_1$.
	The SLP $\H_{\mathsf{L}}$ contains a variable $P_A$ with $\val{P_A} = \rev(u_1 \dots u_m)$.
	Let $\G'$ be the union of $\H_{\mathsf{L}}^\rev$ and $\H_{\mathsf{R}}$
	where $\H_{\mathsf{L}}^\rev$ is obtained from $\H_{\mathsf{L}}$ by reversing all right-hand sides.
	In $\G'$ we view the symbols in $\V$ as variables (they were terminal symbols in $\H_{\mathsf{L}}$ and $\H_{\mathsf{R}}$).
	We remark that $\G'$ may not be contracting caused by heavy symbols $A \in \V$ on a right-hand side.

	We add a rule to $\G'$ for every variable $A \in \V$:
	Let $\beta$ be the root of the heavy tree in $\G$ which contains $A$. We add to $\G'$ the rule
	\begin{equation}
		\label{eq:pref-suf}
		A \to \begin{cases}
			P_A \, \beta \, S_A, & \text{if $\beta \in \Sigma$,} \\
			P_A \, u \, S_A, & \text{if $\beta \in \V$ and $(\beta \to u) \in \R$,}
		\end{cases}
	\end{equation}
	whose right-hand side has length at most $r+2$.
	Using \eqref{eq:deriv} it is easy to show that all variables $A \in \V$ derive the same string in $\G$ and in $\G'$
	by induction on the height of $A$ in $\G$.
	We need to ensure that $\G'$ is contracting.
	Observe that rules of the form \eqref{eq:pref-suf} do not contain a heavy symbol in $u$
	since $u$ originated from a contracting rule $\beta \to u$ of $\G$.
	Hence, the only heavy variables in \eqref{eq:pref-suf} might be $P_A$ and $S_A$.
	In parallel, we expand the heavy variable in all noncontracting rules of the form \eqref{eq:pref-suf}.
	If the heavy symbol is $S_A$ (the case $P_A$ is similar)
	consider the rule for $S_A$ from $\H_{\mathsf{R}}$.
	It is contracting in $\H_{\mathsf{R}}$ but it may contain a heavy variable $B \in \V$,
	which is contained in the light suffix of some ancestor $A'$ of $A$ in $H$.
	In particular $\|B\| \le \|A'\|/2 \le \|A\|/2$,
	and therefore the new rule for $A$ is contracting.
	
	Now all rules for variables $A \in \V$ in $\G'$ are contracting.
	As remarked above, the only remaining heavy variables in $\G'$ are from $\V$
	occurring in rules of $\G'$ originating from $\H_{\mathsf{L}}$ and $\H_{\mathsf{R}}$.
	By expanding such heavy occurrences in parallel we obtain a contracting SLP.
	Since we apply expanding on every rule at most once the right-hand sides of $\G'$ are of length $\O(r)$.
\end{proof}

In the rest of this section we will prove the following result.

\begin{theorem}
	\label{thm:tree-slp}
	Given a labeled tree $T$ with $n$ edges and labels of length $\le \ell$,
	one can compute in linear time a contracting SLP with $\O(n)$ variables
	and right-hand sides of length $\O(\ell)$ defining all prefixes in $T$.
\end{theorem}

\begin{proof}[Proof of \Cref{thm:contracting}]
	We bring the input SLP $\G$ of size $g$ into Chomsky normal form in linear time.
	We compute all left labeled and right labeled heavy trees in $\G$ and
	compute contracting SLPs $\H_{\mathsf{L}}$ and $\H_{\mathsf{R}}$ for their prefixes
	with $\O(g)$ variables and constant-length right-hand sides
	using \Cref{thm:tree-slp}.
	Then \Cref{prop:reduction-to-trees} yields a contracting SLP $\G'$
	which defines all strings that $\G$ defines, has $\O(g)$ variables and constant-length right-hand sides.
\end{proof}

Using the following preprocessing we can always assume that every edge in $T$ is labeled by a single symbol:
Edges labeled by $\varepsilon$ can clearly be contracted.
For every edge labeled by a string $u$ of length $>1$ we replace $u$ by a new symbol $X_u$ of weight $\|u\|$
and introduce the rule $X_u \to u$.
The new goal is to construct an SLP for all prefixes in the new tree with constant-length right-hand sides.
By expanding all heavy occurrences of symbols $X_u$ in such an SLP
we increase the lengths of the right-hand sides by at most $\ell$.
We will also ignore the empty prefix and assume that all symbols in $T$ are distinct.
This allows us to identify the nodes of a derivation tree with their labels, which are variables or terminal symbols.

\subsection{Prefixes of weighted strings}

We start with the case where the tree is a path, i.e.\ we need to define
all prefixes of a weighted string of length $n$ using $\O(n)$ contracting rules.
The following theorem refines \cite[Lemma~III.1]{GanardiJL19} where only the path length from a prefix variable $S_i$
to a symbol $a_j$ in the derivation tree was bounded by $\O(1 + \log \frac{\|S_i\|}{\|a_j\|})$.

\begin{theorem}
	\label{thm:prefix-slp}
	Given a weighted string $s$ of length $n$ one can compute in linear time a contracting SLP with $\O(n)$ variables
	with right-hand sides of length at most 10 that defines all nonempty prefixes of $s$.
\end{theorem}

Let us illustrate the difficulty of defining all prefixes with contracting rules.
Consider the weighted string $s = a_1 \dots a_n$ where symbol $a_i$ has weight $2^{n-i}$.
Since in every factor $a_i \dots a_j$ the left-most symbol $a_i$ is heavy,
every rule for $a_i \dots a_j$ must split off the first symbol $a_i$.
If for every prefix we would only repeatedly split off the first symbol
we would create $\Omega(n^2)$ many variables.
This shows that there is no better solution with right-hand sides of length $\le 2$.
However, using longer rules we can simultaneously reduce both the weight (in a contracting fashion) and the length.

\subparagraph{Computing the base SLP}

Let $s = a_1 \dots a_n$ be a weighted string.
First we recursively construct a contracting ``base'' SLP $\B = (\V,\Sigma,\R,S)$ for $s$.
It will have the additional property of being {\em left-heavy}, i.e.\ for every rule $A \to \beta_1 \dots \beta_k$
and all $2 \le i \le k$ with $\beta_i \in \V$ we have $\|\beta_1 \dots \beta_{i-1}\| \ge \|\beta_i\|$.
Let us emphasize that the condition does not apply when $\beta_i$ is a terminal symbol.
The case $n = 1$ is clear.
If $n > 1$ we factorize $s = u a_i v$ such that $u,v \in \Sigma^*$
have weight at most $\|s\|/2$.
Next factorize $v = v_1v_2$ such that $|v_1|$ and $|v_2|$ differ at most by one.
We add the rule $S \to Ua_iV_1V_2$ to the SLP,
possibly omitting variables if some of the strings $u, v_1, v_2$ are empty.
Finally, we recursively define the variables $U$, $V_1$ and $V_2$.

The SLP $\B$ is clearly contracting, has at most $n$ variables, since every variable can be identified with
the unique symbol $a \in \Sigma$ on its right-hand side, and its right-hand sides have length at most 4.
Notice that the rule $S \to Ua_iV_1V_2$ is left-heavy since $\|ua_i\| > \|s\|/2 \ge \|v_1\|+\|v_2\|$.

\begin{lemma}
	The base SLP $\B$ can be computed in linear time from $s$.
\end{lemma}

\begin{proof}
	For the given weighted string $s = a_1 \dots a_n$ we precompute all prefix sums $\|a_1 \dots a_i\|$,
	which allows us to compute the weight of any factor $\|a_i \dots a_j\|$ in constant time.
	The weighted split point $i$ for $s = a_1 \dots a_n$ can be found using exponential search
	in time $\O(\log (n-i+1))$:
	We compute $\|a_{n-2^k+1} \dots a_n\|$ for $k = 0, 1, \dots$
	until it exceeds $\|s\|/2$, which happens when $k \ge \log(n-i+1)$.
	Then we know that the split point $i$ lies in the interval $[n-2^k+1,n-2^{k-1}]$ of length $2^{k-1}$.
	We can search for the split point in that interval using binary search in time $\O(k) = \O(\log(n-i+1))$.
	Finally we compute the factorization $s = u a_i v_1 v_2$
	and continue recursively with $u$, $v_1$ and $v_2$ (if they are nonempty).
	
	We claim that the running time is $\O(n)$.
	Consider the derivation tree $D$ of $\B$ where each node has up to four children.
	Consider a rule $A \to U a V_1 V_2$ in $\B$ where some of the variables may be missing.
	The time spent at node $A$ is $\O(1 + \log (|V_1|+|V_2|))$.
	We charge cost $\log |V_1|$ to $V_1$ and cost $\log |V_2|$ to $V_2$.
	Notice that the nodes $V_1, V_2$ are always {\em light},
	i.e.\ their leaf size (number of descendant leaves) is at most half the leaf size of their parent node $A$.
	Then the total running time is $\O(n + \sum_A \log |A|)$
	where $A$ ranges over all light nodes in $D$.
	Observe that for every $k$ the set of light nodes $A$ whose leaf size is in $[2^{k-1}+1, 2^k]$
	forms an antichain in the derivation tree and hence there exist at most $n/2^k$ such nodes
	with total cost $\O(k \cdot n/2^k)$.
	The total running time amounts to $\O(\sum_{k=0}^\infty k \cdot n/2^k) = \O(n)$.
\end{proof}

\subparagraph{Defining left branching strings}

Consider the derivation tree $D$ of $\B$.
We identify its nodes with the set of symbols $\S = \V \cup \Sigma$.
Let $\preceq_D$ and $\prec_D$ be the ancestor and the proper ancestor relation on $\S$.
For all $\alpha \preceq_D \beta$ we define $\lft(\alpha,\beta) = u$
where $\alpha \Rightarrow^*_\B u \beta v$ is the unique derivation with $u,v \in \Sigma^*$.
In the derivation tree $\lft(\alpha,\beta)$ is the string that branches off to the left on the path from $\alpha$ to $\beta$.
We have the property
\begin{equation}
	\label{eq:transitive}
	\lft(\alpha,\gamma) = \lft(\alpha,\beta) \, \lft(\beta,\gamma)
\end{equation}
whenever $\alpha \preceq_D \beta \preceq_D \gamma$.
Notice that every proper nonempty prefix of $s$ can be written as $\lft(S,a_i) = a_1 \dots a_{i-1}$.
For all $\alpha \in \S \setminus \{S\}$ with parent node $\alpha'$
we define the {\em left sibling string} $\lsib(\alpha) = u$ where $\alpha' \to u \alpha v$ is the unique rule in $\B$.
It satisfies
\begin{equation}
	\label{eq:lsib-lft}
	\lsib(\alpha) \Rightarrow_\B^* \lft(\alpha',\alpha).
\end{equation}
Notice that we can have $\lft(\alpha,\beta) = \lft(\alpha',\beta')$
for different pairs $(\alpha,\beta)$, $(\alpha',\beta')$.
For a unique description we restrict to a subset of nodes in the derivation tree $D$.
Let $\S_0 \subseteq \S$ be the set of nodes which are not a left-most child in $D$,
i.e.\ symbols $\alpha$ such that $\alpha = S$ or $\lsib(\alpha) \neq \varepsilon$.
In particular the start variable $S$ belongs to $\S_0$.
Observe that $\lft(\alpha,\beta) = \lft(\alpha',\beta')$
where $\alpha'$ and $\beta'$ are the lowest ancestors of $\alpha$ and $\beta$, respectively, that belong to $\S_0$.
In particular, every proper nonempty prefix of $s$ is of the form $\lft(\alpha,\beta)$
for some $\alpha,\beta \in \S_0$.
Let $D_0$ be the unique unordered tree with node set $\S_0$
whose ancestor relation is the ancestor relation of $D$ restricted to $\S_0$.
\Cref{fig:base-slp} shows an example of a tree $D$ with the modified tree $D_0$.

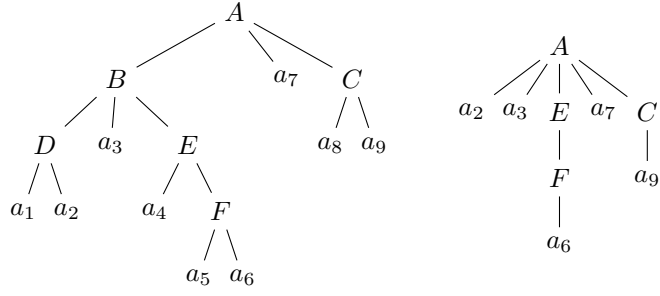
\begin{figure}

\tikzset{every tree node/.style={minimum size = 5pt}}
\tikzset{level distance = 25pt, sibling distance = 0pt}
\tikzset{
edge from parent/.style=
         {draw, edge from parent path={(\tikzparentnode) -- (\tikzchildnode)}}}
\centering
\raisebox{-0.5\height}{
\begin{tikzpicture}
\Tree [.{$A$} [.{$B$} [.{$D$} [.{$a_1$} ] [.{$a_2$} ] ] [.{$a_3$} ]  [.{$E$} [.{$a_4$} ] [.{$F$} [.{$a_5$} ] [.{$a_6$} ] ] ] ]
[.{$a_7$} ]
[.{$C$} [.{$a_8$} ] [.{$a_9$} ] ] ] ]
\end{tikzpicture}
}\hspace{1em}
\raisebox{-0.5\height}{
\begin{tikzpicture}
\Tree [.{$A$} [.{$a_2$} ] [.{$a_3$} ] [.{$E$} [.{$F$} [.{$a_6$} ] ] ] [.{$a_7$} ] [.{$C$} [.{$a_9$} ] ] ]
\end{tikzpicture}
}

\caption{The derivation tree $D$ of a base SLP and the modified tree $D_0$ containing all symbols
which are not a left-most child in $D$.}
\label{fig:base-slp}

\end{figure}

We will introduce variables $L_{\alpha,\beta}$ for the strings $\lft(\alpha,\beta)$.
The variable $L_{\alpha,\beta}$ can be defined using $L_{\alpha',\beta'}$
where $\alpha'$ is a child of $\alpha$ in $D_0$ and $\beta'$ is the parent of $\beta$ in $D_0$.
To achieve the $\O(n)$ bound we will restrict to variables $L_{\alpha,\beta}$
that are used in the derivation of a prefix variable,
namely
\[
	\L = \{ L_{\alpha,\beta} \mid \alpha,\beta \in \S_0, \, \alpha \prec \beta, \, \level(\alpha) \le \height(\beta) \}.
\]
Here $\level(\alpha)$ refers to the length of the path in $D_0$ from the root $S$ to $\alpha$,
and $\height(\beta)$ is the height of the subtree of $D_0$ below $\beta$.

\begin{lemma}
	We can compute in linear time a contracting SLP $\G = (\V \cup \L, \Sigma, \R \cup \Q, S)$
	with right-hand sides of constant length
	such that $\val{L_{\alpha,\beta}} = \lft(\alpha,\beta)$ for all $L_{\alpha,\beta} \in \L$.
\end{lemma}

\begin{proof}
Recall that $\R$ is the set of rules in the base SLP $\B$.
For $\alpha,\beta \in \S_0$ with $\alpha \prec \beta$ and $\level(\alpha) \le \height(\beta)$
we add to $\Q$ the following rule:
\begin{enumerate}[(i)]
\item If $\beta$ is a child of $\alpha$ in $D_0$
then $\lft(\alpha,\beta) = \lft(\alpha',\beta)$ where $\alpha'$ is the parent node of $\beta$ in $D$.
Furthermore $\lsib(\beta) \Rightarrow_\B^* \lft(\alpha',\beta)$ by definition of $\lsib$, and hence
we add the rule $L_{\alpha,\beta} \to \lsib(\beta)$.
\label{item:next-ancestor}
\item If $\beta$ is a child of a child $\gamma$ of $\alpha$ in $D_0$
we add the rule $L_{\alpha,\beta} \to \lsib(\gamma) \, \lsib(\beta)$.
\item Otherwise the path $\pi$ from $\alpha$ to $\beta$ in $D_0$ has length at least 3.
Let $\alpha'$ be the child of $\alpha$ on $\pi$,
and $\beta'$ be the parent node of $\beta$ in $D_0$.
Observe that $\alpha' \prec \beta'$, $\level(\alpha') = \level(\alpha)+1$
and $\height(\beta') \ge 1+\height(\beta)$.
Therefore $\level(\alpha') \le \height(\beta')$.
We introduce the rule
$L_{\alpha,\beta} \to \lsib(\alpha') \, L_{\alpha',\beta'} \, \lsib(\beta)$. \label{item:far-ancestor}
\end{enumerate}
One can prove $\val{L_{\alpha,\beta}} = \lft(\alpha,\beta)$ for all $L_{\alpha,\beta} \in \L$
by induction on the length of the path from $\alpha$ to $\beta$,
using the facts \eqref{eq:transitive} and \eqref{eq:lsib-lft}.
Observe that the right-hand sides of the rules in $\R$ have length at most 7.
However $\G$ is possibly not contracting.

We first show that $L_{\alpha',\beta'}$ is light in the rule
$L_{\alpha,\beta} \to \lsib(\alpha') \, L_{\alpha',\beta'} \, \lsib(\beta)$ from \eqref{item:far-ancestor}.
Let $\alpha''$ be the parent node of $\alpha'$ in $D$
and consider the rule $\alpha'' \to u \alpha' v$ in $\B$.
Observe that $u$ is nonempty since $\alpha' \in \S_0$.
Since $\alpha' \in \V$ is a variable and $\B$ is left-heavy
we know that $\|\alpha'\| \le  \|u\|$, and as $\|\lft(\alpha',\beta')\| \le \|\alpha'\|$ we get $\|\lft(\alpha',\beta')\| \le  \|u\|$.
Then the claim follows from
\[
	2 \cdot \|\lft(\alpha',\beta')\| \le \|u\| + \|\lft(\alpha',\beta')\| = \|\lft(\alpha'',\beta')\| \le \|\lft(\alpha,\beta)\|.
\]
Hence, all heavy variables in $\G$ must be variables $A \in \V$,
which in turn have contracting rules $A \to u$ in $\R$.
By expanding such heavy variables we obtain a contracting SLP
with right-hand sides of length $7-1+4$. This concludes the proof.
\end{proof}

\subparagraph{Size analysis}
We have seen that $\G$ defines all nonempty prefixes
($S$ derives $s$ and every proper nonempty prefix is defined by some variable $L_{S,a_i}$).
To prove \Cref{thm:prefix-slp} it remains to show that $\G$ has $\O(n)$ variables.
We need the following simple lemma.

\begin{lemma}
	\label{lem:louisa}
	Suppose that $T$ is a tree with $n$ leaves such that the height of every subtree $T'$
	is bounded by $\log_2 m$ where $m$ is the leaf size of $T'$.
	Then $\sum_{v \in V(T)} \height(v) = \O(n)$.
\end{lemma}

\begin{proof}
	For $0 \le h \le \height(T)$ let $\#_h(T)$ be the number of nodes in $T$ with height $h$.
	We prove that $\#_h(T) \le n/2^h$ by induction on $n$.
	If $n = 1$ then $T$ only has one node and the statement holds.
	If $n > 1$ consider the subtrees $T_1, \dots, T_k$ rooted in the children of the root of $T$.
	If $h = \height(T) \le \log_2 n$ then $\#_h(T) = 1 \le n/2^h$.
	For $h < \height(T)$ let $n_i$ be the leaf size of the $i$-th child of the root of $T$.
	Then we have $\#_h(T) = \sum_{i=1}^k \#_h(T_i) \le \sum_{i=1}^k n_i/2^h = n/2^h$.
	Therefore
	\[
		\sum_{v \in V(T)} \height(v) = \sum_{h \ge 0} h \cdot \#_h(T) \le n \cdot \sum_{h \ge 0} h \cdot 2^{-h} \le 2n,
	\]
	which concludes the proof.
\end{proof}

\begin{lemma}
	The SLP $\G$ has $\O(n)$ variables.
\end{lemma}

\begin{proof}
	The SLP $\G$ consists of $n$ variables from the base SLP $\B$
	and the variables in $\L$.
	A variable $L_{\alpha,\beta} \in \L$ is uniquely determined by $\beta$ and the level of $\alpha$,
	which is an integer between $0$ and $\height(\beta)$ (height in $D_0$).
	Hence it suffices to show that
	\[
		\sum_{\beta \in \S_0} (\height(\beta) + 1) = \sum_{\beta \in \S_0} \height(\beta) + \O(n) \le \O(n).
	\]
	For the analysis we remove all leaves from $D_0$
	so that all nodes in $D_0$ are variables from $\V$.
	This decreases each term $\height(\beta)$ by exactly one,
	and therefore the sum decreases by at most $|\S_0| = \O(n)$.
	We claim that $|B| \le |A|/2$ whenever $B$ is a child of $A$ in $D_0$.
	Consider the rule $A \to U a_i V_1 V_2$ where some of the variables $U,V_1,V_2$ could be missing.
	Recall that $B$ does not belong to $D_0$ as a left-most child.
	If $B$ is either $V_1$ or $V_2$ then $|A| \ge 1 + (2|B| - 1)$ holds by construction of $\B$.
	If $U \preceq B$ then we obtain $|B| \le |U|/2 \le |A|/2$ inductively.
	Hence $\height(A) \le \log_2 |A|$ and thus, by \Cref{lem:louisa}, we obtain the linear size bound.
\end{proof}

\subsection{Prefixes in trees}

In the light of \Cref{thm:prefix-slp} we can apply \Cref{prop:reduction-to-trees}
to all SLPs $\G$ whose heavy forest is a disjoint union of paths.
This can be easily extended to path-like trees, e.g.\ to {\em caterpillar trees}
where every node has at most one child which is not a leaf.

\begin{proposition}
	\label{prop:path-like}
	Given a labeled caterpillar tree $T$ with $n$ edges and labels of length $\le \ell$,
	one can compute a contracting SLP $\G$ defining all nonempty prefixes in $T$
	such that $\G$ has $\O(n)$ variables and right-hand sides of length $\O(\ell)$.
\end{proposition}

\begin{proof}
	We can assume that every edge is labeled by a single symbol by the preprocessing from the beginning of \Cref{sec:prefix}.
	Consider the path $(v_0, \dots, v_k)$ in $T$ from the root $v_0$ to a leaf $v_k$
	which contains all inner nodes.
	We apply \Cref{thm:prefix-slp} on the string $\omega(v_0,v_1) \dots \omega(v_{k-1},v_k)$
	and obtain a contracting SLP $\G$ of size $\O(n)$
	with variables $P_{v_1}, \dots, P_{v_k}$ producing the prefixes
	$\omega(v_0,v_1) \dots \omega(v_{i-1},v_i)$.
	Any other prefix defined in some node $v$ is of the form $P_{v_i} \omega(v_i,v)$.
	Hence we introduce a rule $P_v \to P_{v_i} \omega(v_i,v)$ and possibly expand the variable $P_{v_i}$,
	if it is heavy.
	The resulting SLP has $\O(n)$ variables and right-hand sides of constant length.
\end{proof}

In the following we construct an SLP for the prefixes in arbitary trees.
The produced SLP will not be contracting in general, but its heavy forest will be a disjoint union of caterpillar trees.
Put differently, such a caterpillar tree consists of a central path $\alpha_1, \dots, \alpha_m$ such that
every $\alpha_i$ occurs at most once heavily in a rule $A \to u$ where $A$ is heavy, namely $A = \alpha_{i-1}$.

\begin{proposition}
	\label{prop:tree-weak}
	Given a labeled tree $T$ with $n$ edges we can compute an SLP $\G$ defining all nonempty prefixes in $T$
	such that
	\begin{enumerate}[(a)]
	\item $\G$ has $4n$ variables and right-hand sides of length $\le 6$,
	\item the subgraph of $\dag(\G)$ induced by the set of heavy symbols is a disjoint union of paths. \label{cond}
	\end{enumerate}
\end{proposition}

\begin{proof}

We proceed by induction on $n$.
The case $n = 1$ is clear so let us assume a tree $T = (V,E,\omega)$ with $n \ge 2$ edges.
A node $v$ in $T$ is {\em unary} if it has exactly one child.
We partition $E$ into maximal {\em unary paths} $\pi = (v_0, \dots, v_k)$, where $k \ge 1$,
$v_1, \dots, v_{k-1}$ are unary nodes, $v_k$ is not unary and $v_0$ is either not unary or the root.
For every such a path $\pi$ we create an SLP $\G_\pi$ containing the rules
\begin{equation}
	\label{eq:rule-0}
	P_{v_0,v_1} \to \omega(v_0,v_1) \quad \text{and} \quad P_{v_0,v_i} \to P_{v_0,v_{i-1}} \omega(v_{i-1},v_i), \text{ for all } 2 \le i \le k.
\end{equation}
We also introduce the abbreviation $P_{v_0,v_0} := \varepsilon$.
Then we construct a new tree $T'$ from $T$
by contracting every maximal unary path $(v_0, \dots, v_k)$ into a single edge $(v_0,v_k)$
labeled by the variable $P_{v_0,v_k}$
defining the string $\omega(v_0,v_1) \dots \omega(v_{k-1},v_k)$ carrying its weight.
This yields a tree where the only possibly unary node is the root.
Therefore it has at least as many leaves as inner nodes.
We then remove all leaves so that $T'$ has at most $n/2$ many edges.
Let $V' \subseteq V$ be the node set of $T'$.

Let $v$ be a node in $T'$.
Let $d(v)$ be the weight of the path from the root to $v$ in $T'$ or equivalently in $T$.
We define the value $\rk(v) = \inf\{ k \in \Z \mid d(v) \le 2^k \}$.
For the root $v$ with $d(v) = 0$ we have $\rk(v) = -\infty$.
Let $\hat v$ be the highest ancestor of $v$ in $T'$ with $\rk(v) = \rk(\hat v)$,
called the {\em peak node} of $v$.
Notice that the root and its children are peak nodes.
Let $Z$ be a maximal set of nodes in $T'$ with the same peak node,
which is a tree rooted in the common peak node, say $\hat v$.
We apply the construction recursively on each tree $Z$.
Let $\G'$ be the union of all obtained SLPs.
Its terminal symbols are of the form $P_{v,v'}$ where $(v,v')$ is an edge in $T'$.
Furthermore, it has at most $4 \cdot n/2 \le 2n$ variables
and the subgraph of $\dag(\G')$ induced by its heavy symbols is a disjoint union of paths.
For every node $v \in V'$ which is not a peak node $\G'$ contains a variable $B_{\hat v, v}$
where $\val{B_{\hat v, v}}$ is the labeling on the path from $\hat v$ to $v$ in $T'$.
To simplify notation we set $B_{\hat v, \hat v} := \varepsilon$ for all peak nodes $\hat v$.

Let $\G$ be the union of $\G'$ and all SLPs $\G_\pi$,
which has at most $n + 2n = 3n$ variables.
For every $x \in V$ which is not the root we add a variable $A_x$
such that $\val{A_x}$ is the labeling of the path from the root to $x$ in $T$.
This yields $4n$ variables, as claimed.
Let $v$ be the lowest ancestor of $x$ in $T$ contained in $V'$.
If $v$ is the root then we add the rule
\begin{equation}
	\label{eq:rule-1}
	A_x \to P_{v,x}.
\end{equation}
Now assume that $v$ is not the root and hence $\hat v$ is also not the root, since the children of the root are peak nodes.
Let $u$ be the parent node of $\hat v$ in $T'$.
If $u$ is the root of $T'$ we add the rule
\begin{equation}
	\label{eq:rule-2}
	A_x \to P_{u,\hat v} \, B_{\hat v, v} \, P_{v,x}.
\end{equation}
Otherwise, $u$ and $\hat u$ are not the root. Let $s$ be the parent node of $\hat u$ in $T'$ and add the rule
\begin{equation}
	\label{eq:rule-3}
	A_x \to A_s \, P_{s,\hat u} \, B_{\hat u, u} \, P_{u, \hat v} \, B_{\hat v, v} \, P_{v,x}.
\end{equation}
One can prove correctness by induction on $\rk(v)$.
Furthermore all new rules have right-hand sides of length $\le 6$.
We claim that the $A$- and $B$-variables are light in the rules
\eqref{eq:rule-2} and \eqref{eq:rule-3}.
Notice that for all $v \in V'$ which are not the root we have
\begin{equation}
	\label{eq:peak}
	\|B_{\hat v, v}\| = d(v) - d(\hat v) < 2^{\rk(v)} - 2^{\rk(v)-1} = 2^{\rk(v)-1} < d(\hat v).
\end{equation}
In the rules \eqref{eq:rule-2} and \eqref{eq:rule-3} the variable $B_{\hat v, v}$ is light
since $2\|B_{\hat v, v}\| < d(\hat v)+\|B_{\hat v, v}\| = \|A_v\| \le \|A_x\|$ by \eqref{eq:peak}.
Similarly, $B_{\hat u, u}$ is light in \eqref{eq:rule-3} because
$2\|B_{\hat u, u}\| < d(\hat u)+\|B_{\hat u, u}\| = \|A_u\| \le \|A_x\|$ by \eqref{eq:peak}.
Finally, $A_s$ is light in \eqref{eq:rule-3} because $\|A_s\| \le 2^{\rk(s)} \le 2^{\rk(v)-2} < \|A_v\|/2 \le \|A_x\|/2$.
This concludes the proof of the claim.

Hence the only possibly heavy symbols in $\G$ are the $P$-variables, the terminal symbols from $\Sigma$
in rules of the form \eqref{eq:rule-0} and the heavy symbols in $\G'$.
It remains to prove \eqref{cond}, i.e.\ every heavy symbol is the heavy child of at most one heavy variable $A$.
Consider a maximal unary path $\pi = (v_0, \dots, v_k)$.
A variable $P_{v_0,v_i}$ where $1 \le i \le k-1$ occurs on the right-hand side of exactly one rule, namely \eqref{eq:rule-0}.
The variable $P_{v_0,v_k}$ can occur on the right-hand sides of \eqref{eq:rule-1}, \eqref{eq:rule-2} and \eqref{eq:rule-3},
but the corresponding left-hand side $A_x$ is not heavy.
By induction hypothesis $P_{v_0,v_k}$ is the heavy child of at most one heavy variable in $\G'$.
Finally, each terminal symbol occurs exactly once on the right-hand side of a rule \eqref{eq:rule-0}.
This concludes the proof.
\end{proof}

\begin{proof}[Proof of \Cref{thm:tree-slp}]
Given a tree $T$ with $n$ edges labeled by single symbols,
we first compute the SLP $\G$ from \Cref{prop:tree-weak}.
Then we compute its left labeled and right labeled heavy trees, which are caterpillar trees
with $\O(n)$ edges and right-hand sides of constant length.
\Cref{prop:path-like} yields contracting SLPs for their prefixes
with $\O(n)$ variables and constant-length right-hand sides.
Hence, by \Cref{prop:reduction-to-trees} we can construct a contracting SLP
defining all strings from $\G$, with $\O(n)$ variables and constant-length right-hand sides.
\end{proof}

\section{Navigation in FSLP-compressed trees}

As a simple application we extend the navigation data structure on FSLP-compressed trees \cite{RehS20}
by the operation which moves to the $i$-th child in time $\O(\log d)$ where $d$ is the degree of the current node.
This is established by applying \Cref{thm:contracting} to the substructure of the FSLP that compresses forests horizontally.

\subparagraph{SLP navigation}
The navigation data structure on FSLPs is based on a navigation data structure on (string) SLPs from \cite{LohreyMR18},
which extends the data structure from \cite{GasieniecKPS05} from one-way to two-way navigation.
The data structure represents a position $1 \le i \le |A|$ in a variable $A$ by a data structure $\sigma(A,i)$,
that we will call {\em pointer},
which is a compact representation of the path in the derivation tree from $A$ to the leaf corresponding to position $i$.
\begin{theorem}[\cite{LohreyMR18}]
\label{thm:slp-navi}
A given SLP $\S$ can be preprocesed in $\O(|\S|)$ time and space so that the following operations are supported in constant time:
\begin{itemize}
\item Given a variable $A$, compute $\sigma(A,1)$ or $\sigma(A,|A|)$.
\item Given $\sigma(A,i)$, compute $\sigma(A,i-1)$ or $\sigma(A,i+1)$, or return $\bot$ if the position is invalid.
\item Given $\sigma(A,i)$, return the symbol at position $i$ in $A$.
\end{itemize}
Furthermore, a single pointer $\sigma(A,i)$ uses $\O(\height(A))$ space and can be computed in time $\O(\height(A))$ for a given pair $(A,i)$.
\end{theorem}

\subparagraph{Forest straight-line programs}
In this section we use the natural term representation for forests.
Let $\Sigma$ be an alphabet of node labels.
The set of {\em forests} is defined inductively as follows:
The concatenation of $n \ge 0$ forests is a forest (this includes the empty forest $\varepsilon$),
and, if $a \in \Sigma$ and $t$ is a forest, then $a(t)$ is a forest.
A {\em context} is a forest over $\Sigma \cup \{x\}$ where $x$ occurs exactly once
and this occurrence is at a leaf node.
If $f$ is a context and $g$ is a forest or a context then $f \langle g \rangle$ is
obtained by replacing the unique occurrence of $x$ in $f$ by $g$.
A {\em forest straight-line program (FSLP)} $\G = (\V_0,\V_1,\Sigma,\R,S)$
consists of finite sets of forest variables $\V_0$ and context variables $\V_1$, the alphabet $\Sigma$,
a finite set of rules $\R$, and a start variable $S \in \V_0$.
The rules contain arbitrary applications of horizontal concatenation and substitutions of forest and context variables.
We restrict ourselves to rules in a certain {\em normal form}, which can be established in linear time with a constant factor size increase \cite{GasconLMRS20}.
The normal form assumes a partition $\V_0 = \V_0^\top \cup \V_0^\bot$
where $\V_0^\bot$-variables produce trees whereas $\V_0^\top$-variables produce forests with arbitrarily many trees.
The rules in $\R$ have one of the following forms:
\begin{align*}
A & \to \varepsilon & \text{ where } & A \in \V_0^\top, & \\
A & \to BC & \text{ where } & A \in \V_0^\top \text{ and } B,C \in \V_0, & \\
A & \to a(B) & \text{ where } & A \in \V_0^\bot, \, a \in \Sigma, \text{ and } B \in \V_0, & \\
A & \to X\langle B \rangle & \text{ where } & A,B \in \V_0^\bot \text{ and } X \in \V_1, & \\
X & \to Y\langle Z \rangle & \text{ where } & X,Y,Z \in \V_1 & \\
X & \to a(L x R) & \text{ where } & X \in \V_1, \, a \in \Sigma, \text{ and } L,R \in \V_0, &
\end{align*}

Every variable $A \in \V_0$ derives a forest $\val{A}$ and
every variable $X \in \V_1$ derives a context $\val{X}$, see \cite{GasconLMRS20} for formal definitions.
An example FSLP for a tree is shown in \Cref{fig:fslp-tree}.
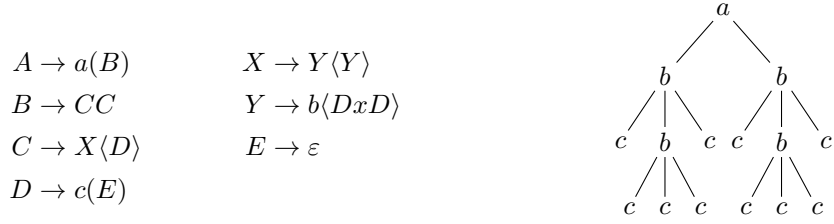
\begin{figure}

\tikzset{every tree node/.style={minimum size = 5pt, inner sep = 2pt}}
\tikzset{level distance = 25pt}
\tikzset{level 2/.style={sibling distance=-5pt}}
\tikzset{level 3/.style={sibling distance=5pt}}
\tikzset{edge from parent/.style={draw, edge from parent path={(\tikzparentnode) -- (\tikzchildnode)}}}
\centering
\begin{minipage}{0.5\textwidth}
\centering
\begin{align*}
A &\to a(B) & X &\to Y\langle Y \rangle \\
B &\to CC & Y &\to b\langle DxD \rangle \\
C &\to X\langle D \rangle & E &\to \varepsilon \\
D &\to c(E)
\end{align*}
\end{minipage}
\begin{minipage}{0.4\textwidth}
\centering
\begin{tikzpicture}
\Tree [.{$a$} [.{$b$} [.{$c$} ] [.{$b$} [.{$c$} ] [.{$c$} ] [.{$c$} ] ]  [.{$c$} ] ] [.{$b$} [.{$c$} ] [.{$b$} [.{$c$} ] [.{$c$} ] [.{$c$} ] ]  [.{$c$} ] ] ]
\end{tikzpicture}
\end{minipage}

\caption{An example FSLP with the variables $\V_0^\bot = \{A,C,D\}$, $\V_0^\top = \{B,E\}$ and $\V_1 = \{X,Y\}$. The tree defined by $A$ is displayed on the right.}
\label{fig:fslp-tree}

\end{figure}

The normal form allows us to define two string SLPs (without start variables) that capture the horizontal and the vertical compression in $\G$.
The {\em rib SLP} $\G_\boxminus = (\V_0,\Sigma_\boxminus,\R_\boxminus)$
over the alphabet $\Sigma_\boxminus = \{ \underline{A} \mid A \in \V_0^\bot \}$
contains all rules of the form $A \to \varepsilon$ or $A \to BC$ from $\R$ where $A \in \V_0^\top$,
and the rule $A \to \underline{A}$ for all $A \in \V_0^\bot$.
We write $\val{A}_\boxminus = \underline{A_1} \dots \underline{A_n}$ for the string derived by $A$ in $\G_\boxminus$,
which satisfies $\val{A} = \val{A_1} \dots \val{A_n}$.
In the example of \Cref{fig:fslp-tree} we have $\val{B}_\boxminus = \underline{C} \, \underline{C}$.
The {\em spine SLP} $\G_\boxbar = (\V_0^\bot \cup \V_1,\Sigma_\boxbar,\R_\boxbar)$
is defined over the alphabet
\[
	\Sigma_\boxbar = \{a(B) \mid (A \to a(B)) \in \R \} \cup \{ a(LxR) \mid (X \to a(LxR)) \in \R \}.
\]
The set $\R_\boxbar$ contains all rules $A \to a(B)$ and $X \to a(LxR)$ from $\R$.
It also contains the rule $A \to X$ for all $(A \to X\langle B \rangle) \in \R$ where $A \in \V_0^\bot$,
and $X \to YZ$ for all $(X \to Y\langle Z \rangle) \in \R$.
We write $\val{V}_\boxbar$ for the string derived by $V$ in $\G_\boxbar$.
If $X \in \V_1$ and $\val{X}_\boxbar = a_1(L_1 x R_1) \dots a_n(L_n x R_n)$ then $\val{X}$ is the vertical composition
of all contexts $a_i(\val{L_i} x \val{R_i})$.
In the example of \Cref{fig:fslp-tree} we have $\val{C}_\boxminus = b\langle DxD \rangle \, b\langle DxD \rangle$.

\subparagraph{FSLP navigation}

Now we define the data structure from \cite{RehS20}.
It represents a node $v$ in a tree produced by a variable $A \in \V_0$ by a pointer $\tau(A,v)$,
which is basically a sequence of navigation pointers in the SLPs $\G_\boxminus$ and $\G_\boxbar$
describing the path from the root of $\val{A}$ to $v$.
Intuitively, the pointer $\tau(A,v)$ can be described as follows.
First we select the subtree of $\val{A}$ which contains $v$, by navigating in $\G_\boxminus$ to a symbol $\underline{B_0}$ where $B_0 \in \V_0^\bot$.
The tree $\val{B_0}$ is defined by a sequence of insertion rules
$B_0 \to X_1 \langle B_1 \rangle, \, B_1 \to X_2 \langle B_2 \rangle, \dots, B_{k-1} \to X_k \langle B_k \rangle$,
where possibly $k = 0$, and a final rule $B_k \to a(C)$.
We navigate in $\G_\boxbar$ in the variable $B_0$ from left to right.
The string $\val{B_0}_\boxbar$ specifies the contexts $a_j(L_j x R_j)$ which together form the context $\val{X_1}$.
If we encounter a context $a_j(L_j x R_j)$ which contains $v$, there are two cases.
If $v$ is the $a_j$-labeled root then we are done.
If $v$ is contained in either $L_j$ or $R_j$ then we record the direction ($\mathsf{L}$ or $\mathsf{R}$)
and continue recursively from the variable $L_j$ or $R_j$.
If $v$ is not contained in the context $X_1$ then we reach the end of $\val{B_0}_\boxbar$,
and continue searching from $B_1$, etc.
If $v$ is contained in $B_k$ then it is either its root or it is contained in $C$.
In the former case, we are done; in the latter case we record the direction $\mathsf{M}$ and continue recursively from $C$.

To define $\tau(A,v)$ formally, let us write $\sigma_\boxminus(A,i)$ and $\sigma_\boxbar(A,i)$ for the pointers to the $i$-th position of a variable $A \in \V_0$
in $\G_\boxminus$ and $\G_\boxbar$, respectively.
We represent every node $v$ in every variable $A \in \V_0$ by a {\em horizontal pointer} $\tau_\boxminus(A,v)$.
Furthermore, we represent every node $v$ in every variable $A \in \V_0^\bot$, deriving a tree,
by a {\em vertical pointer} $\tau_\boxbar(A,v)$.
The pointers are defined recursively as follows:
\begin{enumerate}
\item Let $A \in \V_0$. If $\val{A}_\boxminus = \underline{A_1} \dots \underline{A_n}$ and $v$ is contained in $\val{A_i}$
then set $\tau_\boxminus(A,v) := \sigma_\boxminus(A,i) \, \tau_\boxbar(A_i,v)$.
\item Let $A \in \V_0^\bot$ with a rule $A \to a(B)$. If $v$ is the root of $\val{A}$ set $\tau_\boxbar(A,v) := \sigma_\boxbar(A,1)$,
and otherwise $\tau_\boxbar(A,v) := \sigma_\boxbar(A,1) \, \mathsf{M} \, \tau_\boxminus(B,v)$.
\item Let $A \in \V_0^\bot$ with a rule $A \to X \langle B \rangle$ and
$\val{A}_\boxbar = \val{X}_\boxbar = a_1(L_1 x R_1) \dots a_n(L_n x R_n)$.
If $v$ is contained in $f_i = a_i(\val{L_i} x \val{R_i})$
set $\tau_\boxbar(A,v)$ to be $\sigma_\boxbar(A,i)$, $\sigma_\boxbar(A,i) \, \mathsf{L} \, \tau_\boxminus(L_i, v)$
or $\sigma_\boxbar(A,i) \, \mathsf{R} \, \tau_\boxminus(R_i, v)$, depending whether $v$ is the root of $f_i$
or is contained in $\val{L_i}$ or $\val{R_i}$.
If $v$ is contained in $\val{B}$ set $\tau_\boxbar(A,v) := \sigma_\boxbar(A,n) \, \tau_\boxbar(B,v)$.
\end{enumerate}
For the navigation we only use the horizontal pointers and write $\tau(A,v)$ instead of $\tau_\boxminus(A,v)$.

\begin{theorem}[\cite{RehS20}]
\label{thm:fslp-old-navi}
A given FSLP $\G$ can be preprocesed in $\O(|\G|)$ time and space so that the following operations are supported in constant time:
\begin{itemize}
\item Given a variable $A$, compute $\tau(A,v)$ where $v$ is the root of the first/last tree in $\val{A}$.
\item Given $\tau(A,v)$, compute $\tau(A,v')$ where $v'$ is the parent, first/last child or left/right sibling of $v$,
or return $\bot$ if it does not exist.
\item Given $\tau(A,v)$, return the symbol of node $v$.
\end{itemize}
\end{theorem}

\subparagraph{Navigation to a child}
We extend \Cref{thm:fslp-old-navi} by the operation which, given a pointer $\tau(S,v)$ and a number $1 \le j \le d$,
where $v$ has degree $d$, moves the pointer to the $j$-th child of $v$ in $\O(\log d)$ time.
To this end we apply \Cref{thm:contracting} to $\G_\boxminus$ so that every variable $A \in \V_0$ in the rib SLP has height $\O(\log |\val{A}_\boxminus|)$,
by adding only $\O(g)$ new variables.
In particular, we can compute a pointer $\sigma_\boxminus(A,i)$ in $\O(\log |\val{A}_\boxminus|)$ time by \Cref{thm:slp-navi}.
Furthermore, we compute the length $|\val{A}_\boxminus|$ for all $A \in \V_0$ in linear time.

Suppose we are given a pointer $\tau(S,v)$ to a node $v$ with degree $d$
for some variable $S \in \V_0$.
We show how to compute $\tau(S,v_j)$ where $v_j$ is the $j$-th child of $v$ in $\O(\log d)$ time.
\begin{enumerate}
\item In the first case the last pointer in $\tau(S,v)$ is $\sigma_\boxbar(A,1)$ where the rule of $A \in \V_0^\bot$ is of the form $A \to a(B)$.
Here $B$ derives the forest below the $a$-node and we need to move to the root of the $j$-th tree in the forest.
We compute the pointer $\sigma_\boxminus(B,j)$ in $\O(\log |\val{B}_\boxminus|) \le \O(\log d)$ time.
Then we query the symbol $\underline{B_j}$ at pointer $\sigma_\boxminus(B,j)$ and compute the pointer $\sigma_\boxbar(B_j,1)$ in constant time.
Then we obtain $\tau(S,v_j) = \tau(S,v) \, \mathsf{M} \, \sigma_\boxminus(B,j) \, \sigma_\boxbar(B_j,1)$.
\item In the second case the last pointer in $\tau(S,v)$ is $\sigma_\boxbar(A,i)$ where the rule of $A \in \V_0^\bot$ is of the form $A \to X \langle B \rangle$.
We query the symbol $a_i (L_i x R_i)$ at pointer $\sigma_\boxbar(A,i)$.
The $j$-th child $v_j$ is either in $L_i$, $R_i$ or at the position of the parameter $x$.
If $j = |\val{L_i}_\boxminus| + 1$ we replace $\sigma_\boxbar(A,i)$ by $\sigma_\boxbar(A,i+1)$ in constant time.
If this is not successful then $v_j$ is the root of $B$
and we have $\tau(S,v_j) = \tau(S,v) \, \sigma_\boxbar(B,1)$, which can be computed in constant time.
If $j \le |\val{L_i}_\boxminus|$ we compute $\sigma_\boxminus(L_i,j)$ in $\O(\log |\val{L_i}_\boxminus|) \le \O(\log d)$ time.
We query the symbol $B_j$ at $\sigma_\boxminus(L_i,j)$ and compute $\sigma_\boxbar(B_j,1)$ in constant time.
Then we have $\tau(S,v_j) = \tau(S,v) \, \mathsf{L} \, \sigma_\boxminus(L_i,j) \, \sigma_\boxbar(B_j,1)$.
If $j \ge |\val{L_i}_\boxminus| + 2$ we proceed similarly using $\sigma_\boxminus(R_i,j-|\val{L_i}_\boxminus|-1)$.
\end{enumerate}

\subparagraph{Remarks}
In its original form the SLP navigation data structure from \cite{LohreyMR18} is non-persistent, i.e. the operations modify the given pointer.
However, it is not hard to adapt the structure so that an operation returns a fresh pointer,
by representing paths in the derivation tree using linked lists that share common prefixes.
In a similar fashion, \Cref{thm:fslp} can be adapted so that a pointer is not modified by a navigation step.

Finally, let us comment on the space consumption of a single pointer in \Cref{thm:fslp}.
A single pointer $\tau(A,v)$ consists of a sequence of pointers in $\G_\boxminus$ and $\G_\boxbar$
that almost describes a path in the derivation tree of $A$ in $\G$.
The sequence may contain pointers $\sigma_\boxbar(A,n)$ that point to the lowest node above the parameter
of a context $\val{X}$.
However, in the representation of \cite{LohreyMR18} such a pointer $\sigma_\boxbar(A,n)$ only uses $\O(1)$ space,
since it is a rightmost path in the derivation tree of $A$ in $\G_\boxbar$.
Therefore $\tau(A,v)$ uses $\O(\height(A))$ space where $\height(A)$ is the height of the derivation tree of $A$ in $\G$.
By \cite[Theorem~VII.3]{GanardiJL19} we can indeed assume that the FSLP $\G$ has $\O(\log N)$ height while retaining the size bound of $\O(|\G|)$.
We also need the fact that the transformation into the normal form increases the height only by a constant factor.
However, since the application of \Cref{thm:contracting} to the rib SLP may possibly increase the total height of the FSLP $\G$
by more than a constant factor, it is unclear whether \Cref{thm:fslp} can be achieved with $\O(\log N)$ sized pointers.

\section{Finger search in SLP-compressed strings}

In this section we present our solution (\Cref{thm:finger}) for the finger search problem using contracting SLPs.
Our finger data structure is an {\em accelerated path},
which compactly represents the path from root to the finger in the derivation tree
using precomputed forests on the dag of the SLP.
To move the finger we ascend to some variable on the path,
branch off from the path, and descend in a subtree while computing the new accelerated path.
We can maintain the accelerated path in a dynamic predecessor structure with constant update and query time,
thanks to the $\O(\log N)$ height of the SLP.
We follow the approach of \cite{BilleCCG18} and present an improved $\O(tg)$ space solution for the {\em fringe access problem}:
Given a variable $A$ and a position $1 \le i \le |A|$,
we can access the $i$-th symbol of $\val{A}$ in time $\O(\log d + \log^{(t)} N)$ where $d = \min\{i,|A|-i+1\}$
is the distance from the fringe of $A$, and $t$ is any parameter.

\subsection{Predecessor and weighted ancestor data structures}

\label{sec:ds}

Recall that we assume the word RAM model with word size $w \ge \log N$ where $N$ is the string length.
Since all occurring sets and trees have size $n \le N$ we have $w \ge \log n$ in the following.
We use a dynamic predecessor data structure by P\v{a}tra\c{s}cu-Thorup, which represents a dynamic set $S$
of $n = w^{\O(1)}$ many $w$-bit integers in space $\O(n)$, supporting the following updates and queries in constant time \cite{PatrascuT14}:
$\insert(S,x) = S \cup \{x\}$, $\delete(S,x) = S \setminus \{x\}$, $\pred(S,x) = \max\{y \in S \mid y < x\}$,
$\suc(S,x) = \min\{y \in S \mid y > x\}$, $\rank(S,x) = |\{y \in S \mid y < x\}|$, and $\select(S,i) = x$ with $\rank(S,x) = i$, if any.
By enlarging the word size to $2w$ we can identify a number $x \cdot 2^w + y$, where $x,y$ are $w$-bit numbers,
with the key-value pair $(x,y)$, allowing us to store key-value pairs in the data structure sorted by their keys.
We remark that all standard operations on a $2w$-bit word RAM can be simulated by a constant number of $w$-bit operations.
This dynamic predecessor structure is used to maintain the accelerated path in the derivation tree to the finger.
We extend the data structure by the operation $\split(S,x) = \{y \in S \mid y \le x \}$
for $\O(w)$-sized sets.

\begin{theorem}[\cite{PatrascuT14}]
	\label{thm:dyn-pred}
	There is a data structure representing a dynamic set $S$ of at most $n = \O(w)$ many $w$-bit numbers in space $\O(n)$
	supporting the operations $\insert(S,x)$, $\split(S,x)$ and $\pred(S,x)$ in constant time.
\end{theorem}

\begin{proof}
	Let $T$ be the set of all numbers that have been added to $S$ at some point,
	which is stored in a dynamic predecessor structure with constant update/query time.
	Additionally we maintain a bit vector $b[0..n-1]$
	where $b[i] = 1$ if and only if $i < |T|$ and the $i$-th element of $T$ is contained in $S$
	(we count starting with zero).
	\begin{itemize}
	\item To determine $\pred(S,x)$
	we first compute $i = \rank(T,x)$.
	We can determine the maximal $j \le i$ with $b[j] = 1$
	using a left shift and determining the most significant set bit,
	which can be computed using multiplication \cite{FredmanW93}.
	If there is no such $j$ then $x$ has no predecessor in $S$;
	otherwise the predecessor is $\select(T,j)$.
	\item To insert element $x$ into $S$
	we compute $i = \rank(T,x)$ and check whether $\select(i) = x$.
	If so, $x$ is already contained in $T$ and we set $b[i] = 1$.
	Otherwise, we insert $x$ into $T$ and insert a 1-bit before the $i$-th entry of $b$.
	To prevent $b$ and $T$ from overflowing, before inserting $x$ we delete an element $y \in T \setminus S$ from $T$
	if $|T| = n$ and $|S| < n$.
	Such an element $y \in T \setminus S$ can be found by computing a position $i$ with $b[i] = 0$ and setting $y = \select(T,i)$.
	Then we delete the 0-bit at position $i$.
	All of this can be implemented using a constant number of bit operations
	\item For $\split(S,x)$ we determine $i = \rank(T,x+1)$
	and set bits $i$ to $n-1$ in $b$ to zero,
	by computing $b \wedge \mathtt{0}^{n-i} \mathtt{1}^i$ where $\mathtt{0}^{n-i} \mathtt{1}^i = (1 \ll i) - 1$.
	\end{itemize}
\end{proof}

A {\em weighted tree} $T$ is a rooted tree where each node $v$ carries a nonnegative integer $d(v)$,
called the {\em weighted depth}, satisfying $d(u) \le d(v)$ for all nodes $v$ with parent $u$.
Given a node $v$ and a number $p \in \N$, the {\em weighted ancestor query} $(v,p)$ asks
to return the highest ancestor $u$ of $v$ with $d(u) > p$.
In other words, it is a successor query on the set of ancestors of $v$.
Given a node $v$ and $p \in \N$, we can also compute the highest ancestor $u$ of $v$
where the {\em weighted distance} $d(u,v) = d(v) - d(u)$ is less than $p$, by the weighted ancestor query $(v,d(v)-p)$.
In our application the edges have nonnegative edge weights
and the weighted depth of a node is computed as the sum of all edge weights on the path from the root.

Kopelowitz and Lewenstein \cite{KopelowitzL07} showed how to reduce the weighted ancestor problem to the predecessor problem:
Weighted ancestor queries on a tree of size $n$ can be answered in time $\O(\mathrm{pred}(n) + \log^* n)$
where $\mathrm{pred}(n)$ is the query time of a predecessor data structure.
It was claimed in \cite{GawrychowskiLN14} that the additive term $\log^* n$ can be eliminated without giving an explicit proof.
For completeness we give a proof in \Cref{app:waq} for the setting where $n \le w$
using the predecessor structure from \cite{PatrascuT14}.

\begin{restatable}{proposition}{propwaq}
\label{prop:waq}
A weighted tree $T$ with $n \le w$ nodes can be preprocessed
in $\O(n)$ space and time so that weighted ancestor queries can be answered in constant time.
\end{restatable}

Furthermore, we can also support constant time weighted ancestor queries if the height of the tree
is bounded by $\O(w)$.

\begin{proposition}
	\label{prop:static-wa}
	A weighted tree $T$ with $n$ nodes and height $h = \O(w)$ can be preprocessed in $\O(n)$ space and time
	so that weighted ancestor queries can be answered in constant time.
\end{proposition}

\begin{proof}

We again use the predecessor data structure from \cite{PatrascuT14} on sets of size $h = \O(w)$,
which supports constant time successor queries.
A node is called {\em macro} if it has at least $w$ descendant nodes, and otherwise {\em micro}.
This partitions $T$ into a {\em macro tree}, containing all macro nodes, and multiple {\em micro trees}.
On every micro tree of size $\le w$ we construct the data structure from \Cref{prop:waq},
supporting weighted ancestor queries in constant time.
Since the macro tree has at most $n/w$ many leaves
we can construct a predecessor structure from \cite{PatrascuT14} over every root-to-leaf path in the macro tree.
It contains the pairs $(d(v),v)$ for all nodes $v$ on the path such that no proper ancestor of $v$ has the same weighted depth.
By storing a pointer from each macro node to some predecessor structure for a path containing $v$
we can answer weighted ancestor queries on the macro tree in constant time.
Notice that the total preprocessing time and space is $\O((n/w)w) = \O(n)$.
To answer a weighted ancestor query $(v,p)$ on $T$ where $v$ is a micro node we precompute and
store the lowest macro ancestor $\mathrm{lma}(v)$ of $v$.
If $d(\mathrm{lma}(v)) > p$ then the answer node is contained in the macro tree and can be determined in constant time as described above.
Otherwise, the answer node is contained in the micro tree containing $v$
and can be found in constant time using the data structure from \Cref{prop:waq}.
\end{proof}

\subsection{The fringe access problem}

Consider an SLP $\G$ with the variable set $\V$ containing $g$ variables for a string of length $N$.
By \Cref{thm:contracting} we can ensure in linear time that $\G$ is contracting, has size $\O(g)$ and constant-length right-hand sides.
In particular, every variable $A$ has height $\O(\log |A|)$.
We then bring $\G$ into Chomsky normal form, which preserves the size $\O(g)$ and the logarithmic height subtree property.
We precompute in linear time the length of all variables in $\G$.
To simplify notation we assume that the variables $B$ and $C$ in all rules $A \to BC$ are distinct,
which can be established by doubling the number of variables.
We assign to each edge $e$ in $\dag(\G)$ a {\em left weight} $\lambda(e)$ and a {\em right weight} $\rho(e)$:
For every rule $A \to BC$ in $\G$, the edge $e = (A,B)$ has left weight $\lambda(e) = 0$ and right weight $\rho(e) = |C|$,
whereas the edge $e = (A,C)$ has left weight $\rho(e) = |B|$ and right weight $\rho(e) = 0$.

\subparagraph{Forest data structure}
Let $\F$ be a finite set of subgraphs of $\dag(\G)$ with node set $\V$ where each $F \in \F$ is a forest
whose edges point towards the roots (as for example in the heavy forest).
The forests will be computed later in \Cref{prop:acc}.
For every forest $F \in \F$ we define two edge-weighted versions $F_\mathsf{L}$ and $F_\mathsf{R}$
where the edges inherit the left weights and the right weights from $\dag(\G)$, respectively, yielding $2|\F|$ many weighted forests.
Let $\lambda_F(A)$ and $\rho_F(A)$ be the weighted depths of $A$ in $F_\mathsf{L}$ and $F_\mathsf{R}$, respectively.
In $\O(|\F| \cdot g)$ time we compute for each variable $A$ the weighted depths $\lambda_F(A)$ and $\rho_F(A)$
and the root $\root_F(A)$ of the subtree of $F$ containing $A$.
We write $\lambda_F(A,B)$ and $\rho_F(A,B)$ for the weighted distances between $A$ and $B$
in $F_{\mathsf{L}}$ and $F_{\mathsf{R}}$, respectively.
We preprocess all $2|\F|$ weighted forests in time and space $\O(|\F| \cdot g)$
to support weighted ancestor queries in constant time according to \Cref{prop:static-wa}.
This is possible because the height of the forests is $\O(\log N) = \O(w)$.

\subparagraph{Short and long steps}
We denote by $\langle A,i \rangle$ the state in which we aim to compute a compact representation
of the path from $A$ to the $i$-th leaf in the derivation tree of $A$.
Starting from state $\langle A,i \rangle$ we can take short steps and long steps.
\begin{itemize}
\item A {\em short step} considers the rule of $A$:
If it is a terminal rule $A \to a$ we have found the symbol $a$.
If it is a binary rule $A \to BC$ we compare $i$ with $|B|$:
If $i \le |B|$ then the short step leads to $\langle B,i \rangle$, and otherwise to $\langle C, i - |B| \rangle$.
\item A {\em left long step in $F \in \F$} is possible if $i \le \lambda_F(A) + |\root_F(A)|$.
Put differently, the path from $A$ to $A[i]$ branches off to the left on the path from $A$ to $\root_F(A)$,
or continues below $\root_F(A)$.
We determine the highest ancestor $X$ of $A$ in $F_{\mathsf{L}}$ with $\lambda_F(A,X) < i$
and move to $\langle X, i - \lambda_F(A,X) \rangle$.
Using the weighted ancestor data structure on $F$ the variable $X$ can be determined in constant time.
\item Symmetrically, a {\em right long step in $F$} is possible if $|A| - i + 1 \le \rho_F(A) + |\root_F(A)|$.
Put differently, the path from $A$ to $A[i]$ branches off to the right on the path from $A$ to $\root_F(A)$,
or continues below $\root_F(A)$.
After finding the highest ancestor $X$ of $A$ in $F_{\mathsf{R}}$ with $\rho_F(A,X) < |A|-i+1$
we move to $\langle X, |A|-i+1-\rho_F(A,X) \rangle$.
\end{itemize}
If we take a long step in a forest $F$ then a subsequent short step moves us
from one subtree in $F$ to a different subtree, by maximality of the answer from the weighted ancestor query.

\subparagraph{Accelerated paths}
A sequence of short and long steps is summarized in an accelerated path.
A {\em short edge} is an edge $(A,B)$ in $\dag(\G)$
whereas a {\em long edge} is a triple $(A,F,B)$
such that $F \in \F$ contains a (unique) path from $A$ to $B$.
In the triple $(A,F,B)$ we store only an identifier of $F$ instead of the forest itself.
The left weight and the right weight of a long edge $e = (A,F,B)$ are
$\lambda(e) = \lambda_F(A,B)$ and $\rho(e) = \rho_F(A,B)$, respectively.
An {\em accelerated path} is a sequence $(e_1, \dots, e_m)$ of short and long edges
such that the target node of $e_{k-1}$ is the source node of $e_k$ for all $1 \le k \le m$.
The length of an accelerated path is bounded by $\height(\G) \le \O(\log N)$,
assuming that we do not store long steps from a variable to itself.

\begin{proposition}
	\label{prop:acc}
	Let $t \ge 1$.
	One can compute and preprocess in $\O(tg)$ time a set of forests $\F$ with $|\F| = \O(t)$ so that
	given a variable $A$ and a position $1 \le i \le |A|$,
	one can compute an accelerated path from $A$ to $A[i]$
	in time $\O(\log d + \log^{(t+1)}N)$ where $d = \min\{i, |A|-i+1\}$.
\end{proposition}

\begin{proof}

We first focus on the case that $i \le |A|/2$.
To this end we construct forests $\F = \{F_0, \dots, F_{t-1}\}$ in $\O(tg)$ time
so that an accelerated path from $A$ to $A[i]$ can be computed in time $\O(\log i + \log^{(t+1)}N)$.
As described above, for all $F \in \F$ we construct two constant time weighted ancestor data structures
(for $F_\mathsf{L}$ and $F_\mathsf{R}$),
and compute $\lambda_F(A)$, $\rho_F(A)$ and $\root_F(A)$ for all $A \in \V$.
At the end we comment on the case $i > |A|/2$.

The simple algorithm which only uses short steps takes time $\O(\log |A|)$.
We first improve the running time to $\O(\log i + \log \log |A|)$.
Let $\rk(A) = \min\{ k \in \N : |A| \le 2^{2^k}\}$, which is at most $1 + \log \log |A|$.
The forest $F_0$ contains for every rule $A \to BC$ in $\G$
either the edge $(A,B)$, if $\rk(A) = \rk(B)$,
or the edge $(A,C)$ if $\rk(A) = \rk(C) > \rk(B)$.
If $\rk(A)$ is strictly greater than both $\rk(B)$ and $\rk(C)$ then no edge is added for the rule $A \to BC$.
Observe that any two connected nodes in $F_0$ have the same $\rk$-number.
To query $A[i]$ where $\rk(A) = k$ we proceed as follows:
\begin{enumerate}
\item If $i \le \lambda_{F_0}(A) + |\root_{F_0}(A)|$ we take a left long step in $F_0$ to
some state $\langle X, j \rangle$ with $\rk(X) < k$ and $j \le i$.
We repeat the procedure from there. \label{it:acc-step}
\item Otherwise $i > |\root_{F_0}(A)| > 2^{2^{k-1}} \ge \sqrt{|A|}$.
We query $A[i]$ using short steps in time $\O(\log |A|) \le \O(\log i)$.
\end{enumerate}
Since the rank of the current variable is reduced in every iteration of point \ref{it:acc-step}
and the queried position only becomes smaller
the procedure above takes time $\O(\log i + k) \le \O(\log i + \log \log |A|)$.

We can replace $\log \log |A|$ by $\log^{(t+1)}N$
by adding forests $F_1, \dots, F_{t-1}$ to $\F$:
The forest $F_k$ where $1 \le k \le t-1$
contains for every rule $A \to BC$ in $\G$ either the edge $(A,B)$, if $|B| > \log^{(k)} N$,
or the edge $(A,C)$, if $|B| \le \log^{(k)} N$ and $|C| > \log^{(k)} N$.
Observe that all edges in $F_k$ connect only variables of length $> \log^{(k)} N$.
To query $A[i]$ we compute the maximal $k \in [0,t-1]$ such that $i \le \log^{(k)} N$,
which satisfies either $\log^{(k+1)} N < i$ or $k = t-1$.
\begin{enumerate}
\item If $|A| \le \log^{(k)}N$ we can query $A[i]$ in time
$\O(\log i + \log \log |A|) \le \O(\log i + \log^{(k+2)} N) \le \O(\log i + \log^{(t+1)} N)$.
\item Assume $|A| > \log^{(k)}N$ and hence $k \ge 1$.
Since $i \le \log^{(k)} N < |\root_{F_k}(A)|$ we can take a left long step in $F_k$
and then a short step to some state $\langle X, j \rangle$ where $|X| \le \log^{(k)}N$ and $j \le i$.
We can query $X[j]$ in time $\O(\log j + \log \log |X|) \le \O(\log i + \log^{(k+2)} N) \le \O(\log i + \log^{(t+1)} N)$.
\end{enumerate}
Finally, for every forest $F \in \F$ we include a mirrored version of $F$ which is right-skewed
instead of left-skewed.
This allows us to compute an accelerated path from $A$ to $A[i]$ in time $\O(\log (|A|-i+1) + \log^{(t+1)}N)$,
concluding the proof.
\end{proof}

\subsection{Solving the finger search problem}

We are ready to prove \Cref{thm:finger}.
We maintain an accelerated path $\pi = (e_1, \dots, e_m)$ from the start variable $S$ to the current finger position $f$
with its left weights and right weights as follows.
Let $\ell_j = \sum_{k=1}^j \lambda(e_k)$ and $r_j = \sum_{k=1}^j \rho(e_k)$ be the prefix sums of the weights.
Observe that $f = \ell_m + 1$.
\begin{itemize}
\item We store a stack $\gamma = ((e_1,\ell_1,r_1), (e_2,\ell_2,r_2), \dots, (e_m,\ell_m,r_m))$, implemented as an array.
Given $i \in [1,m]$, one can pop all elements at positions $i+1, \dots, m$ in constant time.
\item We store the set of distinct prefix sums $L = \{ \ell_j \mid 0 \le j \le m \}$
in a dynamic predecessor data structures from \Cref{thm:dyn-pred}
where a prefix sum $\ell$ is stored together with the maximal index $j$ such that $\ell = \ell_j$.
\item Similarly $R = \{ r_j \mid 0 \le j \le m \}$ is stored in a predecessor data structure.
\end{itemize}
For $\setfinger(f)$ we compute an arbitrary accelerated path from $S$ to $S[f]$, say only using only short steps,
and set up the list $\gamma$ and the predecessor data structures for $L$ and $R$ in time $\O(\log N)$.
For $\movefinger(i)$ we can assume that $f-i = d > 0$ since the data structures are left-right symmetric.
By a predecessor query on $L$ we can find the unique index $j$ with $\ell_j < i \le \ell_{j+1}$.
We store $e_{j+1}$ to continue the access to position $i$.
Then we restrict $\gamma$ to its prefix of length $j$,
and perform $\split(L,\ell_j)$ and $\split(R,r_j)$, all in constant time.
Now $\gamma$, $L$ and $R$ represent the accelerated path $(e_1, \dots, e_j)$.

Suppose that $e_{j+1}$ leads from $A$ to $C$.
Its left weight must be positive since $\ell_j < \ell_{j+1}$.
In the following we compute an accelerated path $\pi'$ from $A$ to $S[i] = A[i']$ where $i' = i - \ell_j$
in time $\O(\log d + \log^{(t)} N)$.
We can then compute the new prefix sums, add them to $L$ and $R$, and prolong $\gamma$ by $\pi'$ appropriately.
This is all possible in time $\O(|\pi'|) \le \O(\log d + \log^{(t)} N)$.
It remains to compute the accelerated path $\pi'$ from $A$ to $A[i']$:
\begin{enumerate}
\item If $e_{j+1} = (A,C)$ is a short edge then its left weight is $|B|$ where $A \to B C$ is the rule of $A$.
We take a short step from $\langle A, i' \rangle$ to $\langle B, i' \rangle$,
and compute an accelerated path from $B$ to $B[i']$ in time $\O(\log (|B|-i'+1) + \log^{(t)} N)$ by \Cref{prop:acc}.
This is at most $\O(\log d + \log^{(t)} N)$ because
\[
	|B|-i'+1 = |B| - i + \ell_j + 1 = \ell_{j+1} - i + 1 \le f - i = d.
\]
\item If $e_{j+1} = (A,F,C)$ is a left or right long edge in $F$
we can take a {\em left} long step from $\langle A, i' \rangle$ in $F$ since
\begin{equation}
	\label{eq:i-est}
	i' = i - \ell_j \le \ell_{j+1}-\ell_j = \lambda(e_{j+1}) = \lambda_F(A,C) \le \lambda_F(A).
\end{equation}
The left long step moves from $\langle A, i' \rangle$ to some state $\langle X, i'' \rangle$
where $X$ is the highest ancestor of $A$ in $F$ with $\lambda_F(A,X) < i'$
and $i'' = i' - \lambda_F(A,X)$ in $Y$.
Observe that $X$ is a proper descendent of $C$ in $F$ since $i' \le \lambda_F(A,C)$ by \eqref{eq:i-est}.
In particular there is a binary rule $X \to YZ$ where $Z$ is the parent of $X$ in $F$.
Then we take a short step from $\langle X, i'' \rangle$ to $\langle Y, i'' \rangle$.
We can compute an accelerated path from $Y$ to $Y[i'']$
in time $\O(\log (|Y|-i''+1) + \log^{(t)} N)$ by \Cref{prop:acc}.
This is at most $\O(\log d + \log^{(t)} N)$ because
\[
	|Y|-i'' +1 = |Y| - i' + \lambda_F(A,X) + 1 = \underbrace{|Y| + \lambda_F(A,X)}_{= \lambda_F(A,Z) \le \lambda(e_{j+1})} - i + \ell_j + 1\le f-i = d.
\]
\end{enumerate}
The query $\access(i)$ is similar to $\movefinger(i)$, except that the data structures $\gamma$, $L$ and $R$ are not updated.
This concludes the proof of \Cref{thm:finger}.
We leave it as an open question whether there exists a linear space finger search data structure, supporting $\access(i)$ and $\movefinger(i)$ in $\O(\log d)$ time.
For path balanced SLPs (e.g. $\alpha$-balanced SLPs or AVL-grammars) such a solution does exist.

\begin{theorem}
	\label{thm:opt}
	Given an $(\alpha,\beta)$-path balanced SLP of size $g$ for a string of length $N$,
	one can support $\setfinger(i)$ in $\O(\log N)$ time, and $\access(i)$ and $\movefinger(i)$ in $\O(\log d)$ time,
	where $d$ is the distance between $i$ and the current finger position, after $\O(g)$ preprocessing time and space.
\end{theorem}

\begin{proof}

It suffices to show that fringe access for $(\alpha,\beta)$-balanced SLPs can be solved in time $\O(\log d)$
after linear time preprocessing.
The rest of the proof follows the same route as for \Cref{thm:finger}.
Let $F_0$ and $F_1$ be forests containing all left and right edges:
For every rule $A \to BC$ the forest $F_0$ contains the edge $(A,B)$ and
the forest $F_1$ contains the edge $(A,C)$.
Here it suffices to compute weighted ancestor structures for the right weighted version of $F_0$
and the left weighted version of $F_1$.
To compute an accelerated path from $A$ to $A[i]$ where $i \le |A|/2$
we take a right long step in $F_0$ from $\langle A, i \rangle$ to some state $\langle X, i \rangle$.
If the rule of $X$ is a terminal rule we are done.
Otherwise its rule $X \to YZ$ satisfies $|X| \ge i$ and $|Y| < i$.
Since $\G$ is $(\alpha,\beta)$-path balanced we have $\height(Y) \le \beta \log i$.
Furthermore, since any two root-to-leaf path lengths in the derivation tree below $X$
have a ratio of at most $\beta/\alpha$ we have $\height(Z)+1 \le (\beta/\alpha) (\height(Y)+1)$,
and thus $\height(Z) = \O(\log i)$.
Hence we take a short step to state $\langle Z, i-|Y| \rangle$
and finish the accelerated path in $\O(\log i)$ short steps.
If $i > |A|/2$ we proceed similarly with $F_1$.
\end{proof}

\bibliography{refs}

\appendix

\section{From the weight ancestor problem to the predecessor problem}

\label{app:waq}

\propwaq*

\begin{proof}
First we transform $T$ into a tree $\tilde T$ with pairwise distinct weighted depths.
Let $v_1, \dots, v_n$ be a depth-first traversal of $T'$.
We assign to each node $v_i$ the new weighted depth $\tilde d(v_i) = d(v_i) \cdot 2^w + i$.
Then a weighted ancestor query $(v,p)$ in $T'$ translates into the query $(v,\tilde p)$ in $\tilde T$ where $\tilde p = (p+1) \cdot 2^w$.
We remark that all standard operations on a $2w$-bit word RAM can be simulated by a constant number of $w$-bit operations.
We store all pairs $(\tilde d(v), v)$ in a predecessor data structure $\tilde V$.
Additionally, we store in each node $v$ of $\tilde T$ a bitvector $b(v)$ of length $n \le w$
whose $i$-th bit is one if and only if the node value of $\select(\tilde V, i)$ is an ancestor of $v$.
These bitvectors can be computed in linear time:
The bitvector $b(v)$ can be obtained from the bitvector of its parent by setting the bit at position $\rank(\tilde V, v)$ to one.
To answer a weighted ancestor query $(v,\tilde p)$ in $\tilde T$
we compute $i = \rank(\tilde V, \tilde p)$, compute the minimal $j \ge i$ with $b(v)[j] = 1$
and return the node value of $\select(\tilde V, j)$.
Here the number $j$ is obtained by zeroing out the first $i-1$ many positions in $b$ and computing the most significant bit.
This concludes the proof.
\end{proof}

\end{document}